\documentclass[twocolumn, amsthm]{autart}    % Enable this line and disable the 
\usepackage{graphicx}          % Include this line if your 
\usepackage{xcolor}
\usepackage{amsmath} % add option [fleqn] if want to left-align equations
\usepackage{amssymb}
\usepackage{amsfonts}
\usepackage{graphicx,url} 
\usepackage{bm}
\usepackage{float}
\usepackage{makecell}

\usepackage{enumitem}
\usepackage{subcaption}
\usepackage{wrapfig}
\usepackage{algorithm,algorithmicx} %linesnumbered is an option [ruled], algorithm is a warpper for captions and cross-references
\usepackage{algpseudocode} 
\usepackage{multirow}
\usepackage{diagbox}

\usepackage{cite}
\usepackage{soul}  % for using strike through
\usepackage{xcolor} % before hyperref if you set link colors
\usepackage[bookmarks=false,pdfencoding=auto,psdextra,unicode]{hyperref}
\hypersetup{
    colorlinks=true,
    pdfborderstyle={/S/U/W 1},
    linkcolor=blue,
    filecolor=magenta,      
    urlcolor= blue,
    linkbordercolor=red,% hyperlink border will be red
}
\usepackage[justification=centering]{caption}
\usepackage{subcaption}
\newcommand\rurl[1]{%
  \href{http://#1}{\nolinkurl{#1}}%
}

\usepackage{comment}

\let\abs\relax
\newcommand{\abs}[1]{\left\lvert#1\right\rvert}
\newcommand{\norm}[1]{\left\lVert#1\right\rVert}

\newcommand{\infnorm}[1]{\left\lVert#1\right\rVert_\infty}
\newcommand{\linfnorm}[1]{\left\lVert#1\right\rVert_{\mathcal{L}_\infty}}

\newcommand{\linfnormtruc}[2]{\left\lVert#1\right\rVert_{\mathcal{L}_\infty^{[0,#2]}}}
\newcommand{\lonenorm}[1]{\left\lVert#1\right\rVert_{\mathcal{L}_1}}

\def\lone{{\mathcal{L}_1}}

\def\lonew{${\mathcal{L}_1}$ }

\def\laplace#1{\mathfrak{L}\left[#1\right]}

\def \loneAC {$\lone$AC}

\def\nt {\textup{n}}

\def\tilx{\tilde{x}}

\def \hsigma{\hat{\sigma}}
\def\xin{x_\textup{in}}

\def\xr{x_\textup{r}}

\def \rt {\textup{r}}
\def \rhoin {\rho_\textup{in}}
\def \Zn {\mathbb{Z}_1^n}
\def \Zm {\mathbb{Z}_1^m}

\def \xn {x_\textup{n}}

\def \un {u_\textup{n}}

\def\mbR{\mathbb{R}}

\def\mbZ{\mathbb{Z}}

\def\mbZ{\mathbb{Z}}

\def\hsigma{\hat{\sigma}}

\def\mcC{\mathcal{C}}

\def\mcX{\mathcal{X}}
\def\mcH{\mathcal{H}}
\def\mcU{\mathcal{U}}
\def\mcG{\mathcal{G}}

\def\trieq{\triangleq}

\newtheorem{theorem}{Theorem}
\newtheorem{lemma}{Lemma}

\newtheorem{assumption}{Assumption}

\newtheorem{remark}{Remark} % 

\usepackage{accents}

\makeatletter
\def\cl@part {\@elt {chapter}}
\makeatother

\usepackage{cleveref}

\crefname{equation}{}{} %skip "eq" or "eqs". 
\crefname{lemma}{Lemma}{Lemmas}
\crefname{theorem}{Theorem}{Theorems}
\crefname{table}{Table}{Tables}
\crefname{figure}{Fig.}{Figs.}
\crefname{remark}{Remark}{Remarks}
\crefname{assumption}{Assumption}{Assumptions}
\crefname{section}{Section}{Sections}
\crefname{definition}{Definition}{Definitions}
\crefname{algorithm}{Algorithm}{Algorithms}
\crefname{proposition}{Proposition}{Propositions}

\makeatletter
\renewcommand*\env@matrix[1][\arraystretch]{%
  \edef\arraystretch{#1}%
  \hskip -\arraycolsep
  \let\@ifnextchar\new@ifnextchar
  \array{*\c@MaxMatrixCols c}}
\makeatother

\def\mbZ{\mathbb{Z}}

\def\mcC{{\mathcal{C}}}
\def\mcY{{\mathcal{Y}}}
\def\mcZ{{\mathcal{Z}}}

\def \mcU{{\mathcal{U}}}

\def \loneAC {$\lone$AC}

\def \xcheckin  {{\check x}_\textup{in}}
\newcommand{\Tau}{\mathrm{T}}

\def \xcheckn {{\check x}_\textup{n}}

\def \at {\textup{a}}

\usepackage{booktabs}

\begin{document}
\setlength{\abovedisplayskip}{3pt}
\setlength{\belowdisplayskip}{3pt}
\begin{frontmatter}

\title{Robust Adaptive MPC in the Presence of Nonlinear Time-Varying Uncertainties: An Uncertainty Compensation Approach\thanksref{footnoteinfo}} % Title, preferably not more 

\thanks[footnoteinfo]{This work is supported by the Air Force Office of Scientific Research Grant (AFOSR) Grant AF FA9550-25-1-0274, the National Aeronautics and Space Administration (NASA) under Grant 80NSSC22M0070, and by the National Science Foundation (NSF) under Grants CMMI 2135925, CPS 2311085 and IIS 2331878. The material is this paper was partially presented at the 2024 American Control Conference (ACC), 10-12 July 2024, Toronto, ON, Canada. Corresponding author Pan Zhao.}

\author[1]{Ran Tao}\ead{rant3@illinois.edu},    % Add the 
\author[2]{Pan Zhao}\ead{pan.zhao@ua.edu},               % e-mail address 
\author[3]{Ilya Kolmanovsky}\ead{ilya@umich.edu},  % (ead) as shown
\author[1]{Naira Hovakimyan}\ead{nhovakim@illinois.edu}  % (ead) as shown

\address[1]{University of
Illinois Urbana-Champaign, Urbana, IL 61801, USA}  % Please supply                                              
\address[2]{University of Alabama, Tuscaloosa, AL 35487, USA}             % full addresses
\address[3]{University of Michigan, Ann Arbor, MI 48109, USA}        % here.
% \address[4]{University of
% Illinois at Urbana-Champaign, Urbana, IL 61801, USA}  
          
                             % keyword list or with the 
                                          % help of the Automatica 
                                          % keyword wizard

\begin{abstract}                          % Abstract of not more than 200 words.
This paper introduces an uncertainty compensation-based robust adaptive model predictive control (MPC) framework for linear systems with nonlinear time-varying uncertainties. The framework integrates an $\lone$ adaptive controller to compensate for the matched uncertainty and a robust feedback controller, designed using linear matrix inequalities, to mitigate the effect of unmatched uncertainty on target output channels. Uniform bounds on the errors between the system's states and control inputs and those of a nominal (i.e., uncertainty-free) system are derived. These error bounds are then used to tighten the actual system's state and input constraints, enabling the design of an MPC for the nominal system under these tightened constraints. Referred to as uncertainty compensation-based MPC (UC-MPC), this approach ensures constraint satisfaction while delivering enhanced performance compared to existing methods. Simulation results for a flight control example and a spacecraft landing on an asteroid demonstrate the effectiveness of the proposed framework.
\end{abstract}

\end{frontmatter}

\section{Introduction}\label{sec:introduction}
Real-world systems must often satisfy strict state and input constraints due to actuator limits, safety requirements, efficiency needs, and the necessity for collision and obstacle avoidance. Model Predictive Control (MPC) has become a leading control approach to handling constraints. MPC optimizes the behavior of a system over a finite horizon that is predicted by a model, while integrating constraints directly into the optimization process \cite{camacho2013mpc-book,rawlings2020mpc-book}. 
Though successfully deployed in many real world systems already, in more demanding applications, MPC must further account for model uncertainties, such as unknown parameters, unmodeled dynamics, and external disturbances, to ensure safety and enhance performance.
Although MPC, if properly formulated, possesses inherent robustness properties to small uncertainties \cite{yu2014inherent}, larger uncertainties need to be explicitly handled in the design. In particular, a variety of {\it robust MPC} have been proposed, see, e.g, {\cite[Section~3]{mayne2014model}}. Two major classes of robust MPC approaches are min-max MPC \cite{scokaert1998min,kerrigan2001robust-mpc} and tube-based MPC \cite{langson2004robust-mpc-tube,rakovic2005robust-mpc,mayne2006robust-mpc,mayne2011tube-mpc-nonlinear,lopez2019dynamic-tube-mpc, limon2008design}. Min-max MPC relies on min-max optimization of the control performance under the worst-case uncertainty to ensure robust constraint enforcement.
Tube-based MPC, on the other hand, relies on an ancillary feedback law $u = \bar{u} + K(x - \bar{x})$ to mitigate uncertainties and computes a ``tube", i.e., a disturbance invariant set, around the uncertainty-free nominal system trajectory, within which the system trajectory remains in the presence of uncertainties. The choice of $K$ crucially determines the size of this ``tube" and, consequently, also the conservativeness and feasibility of the resulting tightened constraints used in the nominal MPC problem. 
Early tube MPC formulations~\cite{rakovic2005robust-mpc,mayne2006robust-mpc} typically leveraged a fixed pre-stabilizing gain (e.g., an LQR gain) that renders the closed-loop dynamics stable, and assumed the uncertainty and the corresponding disturbance invariant set $Z$ are small enough compared with the original state and control input constraints. 
Later works, including \cite{limon2008design}, introduced systematic approaches to synthesize $K$.
In particular,~\cite{limon2008design} formulated this synthesis as a convex optimization problem using Linear Matrix Inequalities (LMI), where the gain $K$ and an invariant ellipsoid $Z$ are jointly optimized to explicitly minimize the size of the invariant set while guaranteeing that the tightened state and input constraint sets $\mathbb{X} \ominus Z$ and $\mathbb{U} \ominus KZ$ remain non-empty.
This LMI-based design ensures both robust constraint enforcement and feasibility of the nominal MPC constraints. 
A notable issue with the robust MPC approaches above is that they accommodate the worst-case uncertainties, and can often result in conservative control performance. This limitation can be mitigated by {\it adaptive MPC}, which builds upon robust MPC methods and reduces conservativeness through online identification of a set that captures the unknown parameters, ensuring robust constraint satisfaction \cite{adetola2011robust,sasfi2023robust,lorenzen2017adaptive}. Such methods usually require an a priori known \textit{parametric} structure for the uncertainties and, moreover, need the parameters to be \textit{time-invariant}. An adaptive MPC framework that accounts for nonlinear time-varying uncertainties was proposed in \cite{wang2017adaptive}, which incorporates online estimation and prediction of uncertainty sets to achieve improved performance compared with traditional robust MPC methods. However, because it remains rooted in min–max formulations, the closed-loop performance has room for further improvement.

A different approach to constrained control in the presence of uncertainty is \emph{uncertainty-compensation-based MPC}, which integrates uncertainty estimation and cancellation with MPC. In this paradigm, the control law includes two components: an auxiliary adaptive controller to compensate for the effect of uncertainties so that the actual system behaves close to an uncertainty-free nominal model, and an MPC with suitably tightened bounds to achieve desired performance and ensure closed-loop constraint enforcement. 
For example, \cite{sinha2022adaptive} proposed a learning-based robust predictive control algorithm that employs online function approximation to compensate for matched uncertainties and uses robust MPC with tightened constraints to address the residual uncertainties. However, the formulation in \cite{sinha2022adaptive} relies on a known parametric structure of the unknown dynamics, which limits generality. More critically, that approach remains conservative before a good model is learned to approximate the uncertainty, since the robust MPC design considers a worst-case scenario.
In another line of work, \cite{pereida2021robust} integrated an \loneAC, as described in \cite{hovakimyan2010L1-book}, with robust MPC for systems with matched uncertainties only. In \cite{pereida2021robust}, the \loneAC~cancels the matched uncertainties and provides uniform bounds on the state and input deviations between the compensated system and the nominal system. The authors of \cite{pereida2021robust} treated the state deviation as a bounded model uncertainty and leveraged robust MPC to address this issue. However, the uncertainty modeled in robust MPC represents dynamic uncertainty, which is fundamentally different from the state deviation bound provided by \loneAC. This conceptual inconsistency raises questions about the validity/appropriateness of the formulation in \cite{pereida2021robust} and highlights the need for a consistent and theoretically grounded framework that rigorously integrates \loneAC~with MPC.

A more systematic example of uncertainty-compensation-based constrained control is given in \cite{zhao2023integrated}, where an integrated reference governor and adaptive control architecture were proposed for linear systems with nonlinear time-varying matched uncertainties. The authors of \cite{zhao2023integrated} leveraged an \loneAC~to compensate for the matched uncertainties and established uniform bounds on the state and input deviations between the actual and nominal systems. These bounds are then used to tighten the state and input constraints of the actual system for reference command modification using the nominal system, ensuring closed-loop constraint satisfaction.
As compared to MPC, the reference governor is a simpler but more restrictive solution that adjusts the reference command of a nominal stabilizing controller to a safe reference command while assuming that the reference command is constant over the prediction horizon. The MPC, on the other hand, manipulates the full control input sequence, without requiring  a nominal stabilizing controller or assuming a  constant control input over the prediction horizon, by solving a dynamic optimization problem subject to constraints.
Motivated by the idea of \cite{zhao2023integrated}, a preliminary version of this work in \cite{tao2024robust} extended the
formulation in \cite{zhao2023integrated} to the MPC setting with both matched and unmatched uncertainties. An \loneAC~was again employed to cancel the matched component and to derive uniform error bounds under both types of uncertainty; these bounds were subsequently incorporated into the nominal MPC via constraint tightening.
However, \cite{tao2024robust} suffers from two major limitations. 
First, in \cite{tao2024robust}, we only analyzed the effect of unmatched uncertainties {\it without actively mitigating them}, which can lead to conservative uniform bounds. Since constraint tightening depends on these bounds, the lack of active mitigation results in conservative error estimates, which shrinks the feasible region of the nominal MPC, leading to degraded performance and potentially infeasible constraint sets. 
Second, in \cite{tao2024robust}, we did not provide guarantees of recursive feasibility or closed-loop stability, which are central to any MPC formulation. 

\textbf{Contributions}: This paper introduces a robust adaptive MPC framework based on uncertainty compensation, referred to as UC-MPC, for linear systems affected by nonlinear time-varying uncertainties that include both matched and unmatched components. 
The proposed framework integrates an \loneAC~to estimate and compensate for the matched uncertainties and a robust feedback controller, designed using LMI, to mitigate the influence of unmatched uncertainties on target output channels. To the best of our knowledge, this is the first framework that simultaneously mitigates both matched and unmatched nonlinear time-varying uncertainties, without imposing structural assumptions, within an MPC formulation. 
The resulting composite control law also reveals a structural connection to tube MPC~\cite{mayne2005robust}.
In addition, we introduce a scaling mechanism for unmatched uncertainties, enabling tighter deviation bounds and reducing constraint tightening. Together with the LMI-based feedback controller, the scaling mechanism enables the use of \loneAC~to derive meaningful uniform bounds on the deviations between the closed-loop system and the nominal uncertainty-free system, which are then embedded into a constraint-tightened MPC design for the nominal system. Moreover, we establish comprehensive theoretical guarantees for UC-MPC, including conditions for constraint satisfaction, performance, and most importantly, the recursive feasibility of the MPC problem and stability of the nominal closed-loop system. Finally, we validate UC-MPC through two case studies, on aircraft flight control and spacecraft soft landing under uncertain asteroid gravity fields, which demonstrate its ability to ensure both safety and efficiency in challenging uncertain environments.

{\it Notations}: Throughout the paper, $\mathbb{R}$, $\mathbb{R}_+$ and $\mbZ_+$ denote the set of real, non-negative real, and non-negative integer numbers, respectively. We use $\mathbb{R}^n$ and  $\mathbb{R}^{m\times n}$ to represent the $n$-dimensional real vector space and the set of real $m$ by $n$ matrices, respectively. The integer sets $\{i, i+1, \cdots\}$ and $\{1, 2,\cdots,n\}$ are denoted by $\mbZ_i$ and $\mbZ_1^n$, respectively. 
$I_n$ denotes a size $n$ identity matrix, and $0$ is a zero matrix of a compatible dimension.
$\norm{\cdot}$ and $\norm{\cdot}_\infty$ represent the $2$-norm and $\infty$-norm of a vector or a matrix, respectively. 
The $\mathcal{L}_\infty$- and truncated $\mathcal{L}_\infty$-norm of a function $h:\mathbb{R}_+ \rightarrow\mathbb{R}^n$ are defined as $\norm{h}_{\mathcal{L}_\infty}\triangleq \sup_{t\geq 0}\infnorm{h(t)}$ and $\linfnormtruc{h}{T}\triangleq \sup_{0\leq t\leq T}\infnorm{h(t)}$, respectively. The $\mathcal{L}_1$ norm of a transfer function $G(s)$ is defined as $\|G(s)\|_{\mathcal{L}_1} \trieq \|g\|_{\mathcal{L}_1} = \int_0^\infty \|g(t)\|dt$ where $g(t)$ is the impulse response of $G(s)$.
The Laplace transform of a function $x(t)$ is denoted by $x(s)\triangleq\mathfrak{L}[x(t)]$.
For any vector $x$, $x_i$ denotes the $i$th element of $x$. Given a positive scalar $\rho$, $\Omega(\rho)\trieq \{z\in \mbR^n: \infnorm{z}\leq \rho \}$ represents a high dimensional ball set of radius $\rho$ which centers at the origin with a compatible dimension $n$. For a high-dimensional set $\mcX$, the interior of $\mcX$ is denoted by $\textup{int}(\mcX)$ and the projection of $\mcX$ onto the $i$th coordinate is represented by $\mcX_i$. For given sets $\mcX,\mcY\subset \mbR^n$, $\mcX \oplus \mcY \trieq \{ x+y: x\in \mcX, y\in \mcY \}$ is the Minkowski set sum and $\mcX\ominus \mcY \trieq \{z: z+y\in \mcX, \forall y\in \mcY \}$ is the Pontryagin set difference. For any matrix $A$, $A^\top$ represent its transpose. 
\vspace{-3mm}

\section{Problem Statement}\vspace{-2mm}
Consider the following linear system with uncertainties 
\begin{equation}\label{eq:dynamics-uncertain-original}
\left\{ \begin{aligned}
  \dot x(t) &= Ax(t) + B(u(t) +  f(t,x(t)))+B_u w(t), \hfill \\
  y(t) &= Cx(t), \ x(0) = x_0,\\ 
\end{aligned}\right.    
\end{equation}
where $x(t)\in\mbR^n$, $u(t)\in \mbR^m$ and $y(t)\in\mbR^p$ represent the state, input, and output vectors, respectively, $x_0\in\mbR^n$ is the initial state, and matrices $A$, $B$, $B_u$, and $C$ are known with compatible dimensions. The matrix $B$ has full column rank, and $B_u$ is a matrix such that $\textup{rank}[B\  B_u]= n$  and $B_u^\top B =0 $. 
%The existence of $B_u$ makes control input $u$ not able to cancel any part of the uncertainty $w$. Moreover,
The function $f(t,x(t))\in \mbR^m$ represents the matched uncertainty dependent on both time and states, and $w(t)\in \mbR^{n-m}$ denotes the unmatched disturbance. Here, we assume full state availability for feedback use.

\begin{remark}[Matched and unmatched uncertainties]
\label{remark:B}
Consider an uncertain system represented by
\[
\dot{x}(t) = A x(t) + B u(t) + f_{\text{total}}(t,x), 
\]
where  $f_{\text{total}}(t,x)$ denotes the total uncertainties and disturbances. 
Since $B$ has full column rank, any vector $f_{\text{total}}(t,x)$ can be decomposed uniquely into two orthogonal components: one lying in the column space of $B$ (matched) and the other in its orthogonal complement (unmatched). Formally,
\[
f_{\text{total}}(t,x) = B f_m(t,x) + B_u f_{um}(t,x),
\]
where $B_u$ spans the null space of $B^\top$ (so that $B_u^\top B = 0$ and $\textup{rank}[B\  B_u]= n$). The matched component can be obtained by projection
\[
f_m(t,x) = B^\dagger f_{\text{total}}(t,x),
\]
where $B^\dagger=(B^\top B)^{-1}B^\top$ is the pseudo-inverse of $B$,
and the unmatched component is given by  $B_u^\dagger f_{\text{total}}(t,x)$. %the residual $f_{\text{total}} - B f_m$. 
This decomposition is not an additional modeling restriction, but rather a consequence of linear algebra given the nominal pair $(A,B)$. %It shows that the system class considered in this work is general: 
The total uncertainty $f_{\text{total}}(t,x)$ can be arbitrary in structure (e.g., nonlinear, time-varying, state-dependent), and the decomposition always exists. The matched component represents the portion of uncertainty that can be directly canceled by control inputs, while the unmatched component cannot be canceled. Our framework explicitly accounts for both components, thereby encompassing a broad class of linear systems with nonlinear, time-varying matched and unmatched uncertainties.% studied in robust adaptive MPC.
\end{remark}

\begin{assumption}\label{assump:lipschitz-bnd-fi}
Let $\mcZ \subset \mathbb{R}^n$ be a compact set. We assume:
\vspace{-4mm}
\begin{enumerate}[label=(\roman*)]
\item  
For each component $f_j(t,x)$ of $f(t,x)$, $j \in \{1,\dots,m\}$, there exist positive constants $L_{f_j,\mcZ}$, $l_{f_j,\mcZ}$, and $b_{f_j,\mcZ}$, such that for all $x,z \in \mcZ$ and all $t,\tau \geq 0$, we have
\begin{subequations}\label{eq:lipschitz-cond-and-bnd-fi}
\begin{align}
\hspace{-8mm}\abs{f_j(t,x) \!-\! f_j(\tau,z)}  & \!\le\! L_{f_j,\mcZ}\infnorm{x \!-\! z} \!+\! l_{f_j,\mcZ}\abs{t-\tau} , \label{eq:lipschitz-cond-fi}\\
\abs {f_j(t,x)}  &\le b_{f_j,\mcZ}, \label{eq:bnd-fi} 
\end{align}
\end{subequations}
which means that each component of $f$ is uniformly Lipschitz continuous in $(t,x)$ and uniformly bounded on $[0,\infty)\times \mcZ$.

\item 
For each component $w_j(t)$ of $w(t)$, $j \in \{1,\dots,n-m\}$, there exists a constant $b_{w_j} > 0$ such that
\begin{equation}\label{eq:bnd-wi}
    \|w_j(t)\|_\infty \le b_{w_j}, 
    \qquad \forall\, t \geq 0.
\end{equation}

\end{enumerate}
\end{assumption}

\begin{remark}
It is possible to extend the results to allow $w(t)$ to be state-dependent, $w(t,x)$, provided the boundedness assumption \eqref{eq:bnd-wi} holds. To simplify the exposition, we limit the derivations in this paper to time-dependence only. 
\end{remark}\vspace{-2mm}
Based on \cref{assump:lipschitz-bnd-fi}, it follows that for any $x,z \in \mcZ$ and $t,\tau\geq 0$, we have
\begin{subequations}\label{eq:lipschitz-cond-bnd-f}
\begin{align}
\infnorm{f(t,x) \!-\! f(\tau,z)}  & \!\le\!{L_{f,\mcZ}}\infnorm{x\! -\! z} \!+\! l_{f,\mcZ}\abs{t-\tau}, \label{eq:lipschitz-cond-f}\\
  \infnorm {f(t,x)}  &\le {b_{f,\mcZ}},  \label{eq:bnd-f} \\
  \infnorm {w(t)}  &\le b_{w}, \label{eq:bnd-w} 
\end{align}
\end{subequations}
where 
\begin{align}
\label{eq:Lf-lf-bf-defn}
        L_{f,\mcZ} &= \max_{j\in \Zm} L_{f_j,\mcZ},\; l_{f,\mcZ} = \max_{j\in \Zm} l_{f_j,\mcZ},\nonumber \\ 
        b_{f,\mcZ} &= \max_{j\in \Zm} b_{f_j,\mcZ}, \;
        b_w = \max_{j\in \mbZ_1^{n-m}} b_{w_j}.
\end{align}

\begin{remark} 
Unlike traditional adaptive control approaches, we specifically make assumptions on $f_j(t,x)$ and $w_j(t)$ instead of on $f(t,x)$ and $w(t)$ as in \cref{eq:Lf-lf-bf-defn} in order to derive an  {\it individual} bound on each state and on each  component of the adaptive control  input (see \cref{sec:l1acbounds} for details).
\end{remark}\vspace{-2mm}
Given the system \cref{eq:dynamics-uncertain-original}, \textbf{the goal of this paper} is to design a robust adaptive MPC scheme such that 
\vspace{-3mm}
\begin{enumerate}[topsep = 0pt,label = (\alph*)]
    \item[(O1)] state and control constraints are always satisfied,
\begin{equation}\label{eq:constraints}
    \begin{gathered}
  x(t) \in \mathcal{X},\quad u(t) \in \mathcal{U}, \quad \forall t\geq 0,\hfill 
\end{gathered} 
\end{equation}
where $\mcX\subset \mbR^n$ and $\mcU \subset\mbR^m$ are pre-specified convex and compact sets including the origin,

\item[(O2)] transient performance bounds for the states and outputs can be established to quantify the control performance, and
\item[(O3)] the effect of uncertainties $f(t,x)$ and $w(t)$ on the control performance is minimized.
\end{enumerate}
\vspace{-2mm}

\section{Overview of the UC-MPC Framework}\label{sec:overview}
\vspace{-2mm}
\begin{figure}[!b]
\begin{center}
     \includegraphics[width = 1\columnwidth]{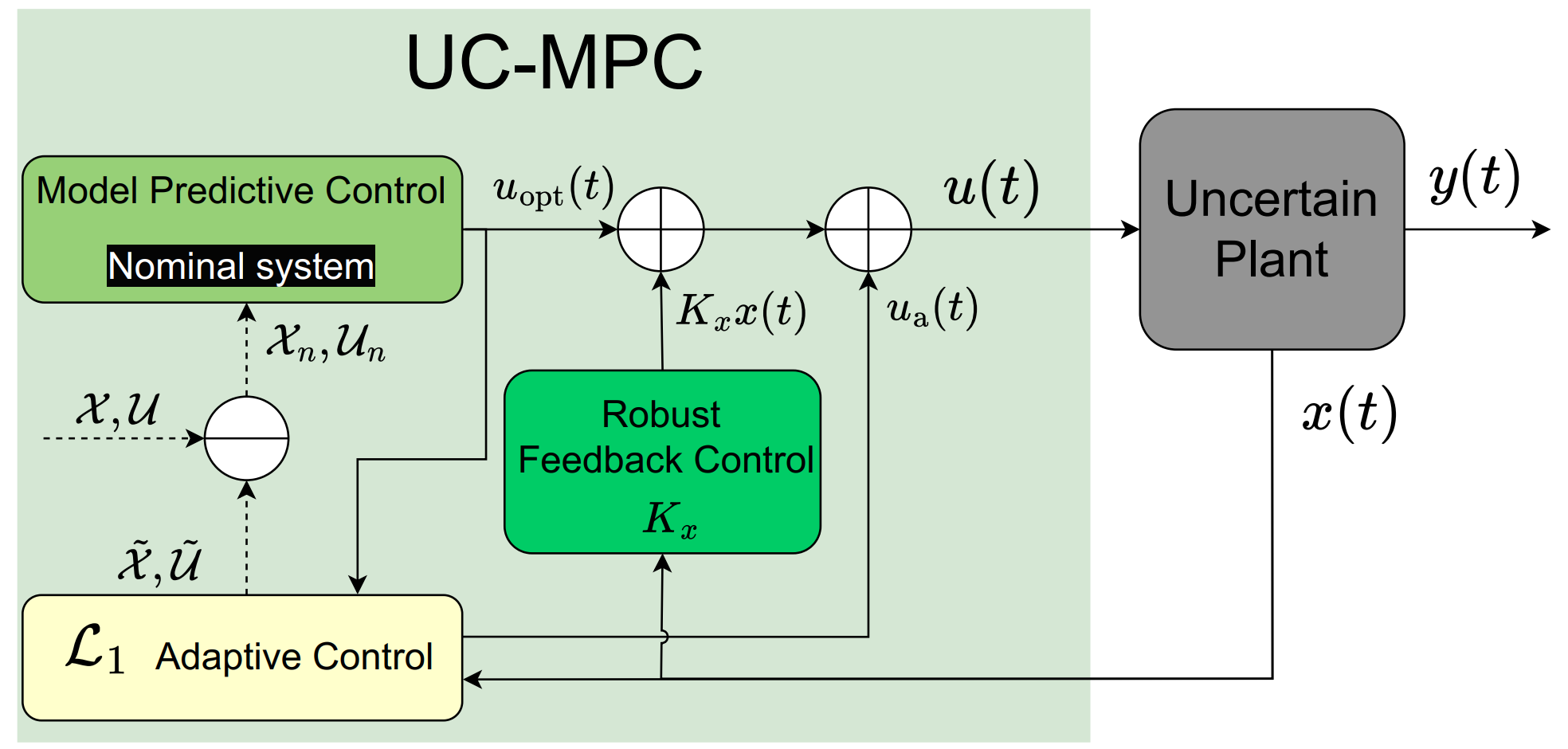}
     \vspace{-4mm}
     \caption{Diagram of the proposed UC-MPC framework. UC-MPC leverages an \loneAC~to estimate and compensate for the matched uncertainty and a robust feedback controller to mitigate the effect of the unmatched uncertainty.}
     \label{fig:framework}
     \vspace{-3mm}
 \end{center}   
 \end{figure}
We present the proposed UC-MPC framework in \cref{fig:framework}. It consists of a robust feedback control law $K_x x(t)$, an \loneAC~to generate $u_a(t)$, and an MPC design for the nominal system with tightened constraints $\mcX_{\nt}$ and $\mcU_{\nt}$ to generate $u_{\text{opt}}(t)$.
In the presence of the uncertainties $f(t,x)$ and $w(t)$, we leverage the feedback control $K_xx(t)$ to mitigate the effect of $w(t)$ on target output channels and an \loneAC~to generate $u_a(t)$ to compensate for the matched uncertainty $f(t,x)$, and to force the actual system to behave close to a nominal (i.e., uncertainty-free) system \cite{hovakimyan2010L1-book}. The designs of LMI-based $K_x$ and \loneAC~are discussed in \cref{sec:kx_design} and \cref{sec:l1acbounds}, respectively. 
Thus, \textbf{the overall control law for the system in \cref{eq:dynamics-uncertain-original} is summarized as}:
\begin{equation}\label{eq:total-control-law}
    u(t)=K_xx(t)+u_{\text{opt}}(t)+u_a(t),
\end{equation}
and we can reformulate the system \cref{eq:dynamics-uncertain-original} as
% \begin{equation}\label{eq:dynamics-uncertain}
% \left\{
% \begin{aligned}
%  \!\dot x(t) \!&=\! {A_m}x(t)\! + \!Bu_{\text{opt}}(t) \!+\! B(u_{\at}(t) \!+\! f(t,\!x(t)\!)\!)\!+\!B_uw(t) \\ 
%  \! y(t) & = Cx(t), \ x(0) =x_0, 
% \end{aligned}\right. 
% \end{equation}
\begin{align}
\!\dot x(t) \!&=\! {A_m}x(t)\! + \!Bu_{\text{opt}}(t) \!+\! B(u_{\at}(t) \!+\! f(t,\!x(t)\!)\!)\!+\!B_uw(t), \nonumber \\
y(t) &= Cx(t), \quad x(0) = x_0,
\label{eq:dynamics-uncertain}
 \end{align}
where ${A_m} \trieq A + B{K_x}$ is a Hurwitz matrix. 
In \cref{sec:l1acbounds}, we will demonstrate that \loneAC~provides uniform bounds on the errors between the states and control inputs of the real plant \cref{eq:dynamics-uncertain} and those of a nominal system below: 
\begin{equation}\label{eq:nominal system}
\left\{
\begin{aligned}
  \dot x_{\nt}(t) &= {A_m}\xn(t) + Bu_{\text{opt}}(t), \ \xn(0) =x_0, \hfill \\
     u_{\nt}(t)& =K_x \xn(t) + u_{\text{opt}}(t),
\end{aligned}\right.
\end{equation}
where $x_{\nt}$ and $u_{\nt}$ are the vectors of nominal states and inputs. To ensure robust constraint satisfaction of system \eqref{eq:dynamics-uncertain} with constraints in \eqref{eq:constraints}, we design an MPC for the nominal system with tightened constraints derived using the uniform bounds from \loneAC. Specifically, the uniform bounds enable us to achieve 
\begin{equation}\label{eq:x-xn-in-X-u-un-in-U}
    \begin{aligned}
    x(t)-\xn(t) \in \tilde {\mcX},  % \trieq \{z\in\mbR^n: \infnorm{z}\leq \alpha_x\}, \\
  \quad  u(t)- \un(t) \in  \tilde {\mcU}, \quad \forall t\geq 0,%\trieq \{z\in\mbR^m: \infnorm{z}\leq \alpha_u\},
    \end{aligned}
\end{equation}
where \( u(t) \) is defined in \cref{eq:total-control-law} and \( \un(t) \) is given in \cref{eq:nominal system}. The sets \( \tilde {\mcX} \) and \( \tilde {\mcU} \) are pre-computed hyperrectangular bounds derived from the bounds on the uncertainties \( f(t,x) \) and \( w(t) \), as well as the design parameters of \loneAC. These bounds allow us to define tightened constraints for the nominal system in \cref{eq:nominal system} as:
\begin{equation}\label{eq:Xn-Un-defn}
    \mcX_{\nt} \trieq \mcX \ominus \tilde {\mcX}, \quad \mcU _{\nt} \trieq \mcU \ominus \tilde {\mcU}.
\end{equation}
While constraint tightening ensures robust constraint satisfaction for the real system, it introduces trade-offs. Tightened constraints reduce the feasible region for the nominal MPC, leading to increased conservativeness. This conservativeness can limit the controller’s ability to fully utilize system resources, potentially resulting in degraded performance. To mitigate this, it is essential to derive meaningful performance bounds from \loneAC~that allow for effective yet minimally conservative constraint tightening.
In this paper, we propose three techniques to achieve meaningful performance bounds. Firstly, we design the feedback gain $K_x$ in \cref{sec:kx_design} to minimize the effect of the unmatched uncertainty on safety-critical or performance-critical state channels. Secondly, we introduce individual performance bounds in \cref{sec:individual bounds} for each state and input channel. Finally, we employ a scaling technique for unmatched uncertainties in \cref{sec:scaling of w} to further refine the performance bounds, ensuring less tight constraints are applied. These techniques collectively enable a more balanced and efficient trade-off between robustness and performance.
\vspace{-3mm}
\subsection{Equivalent Nominal System and Connection with Tube MPC}
\label{sec:tube mpc}
\vspace{-3mm}
The nominal dynamics in \eqref{eq:nominal system} can be equivalently expressed as  
\begin{equation}\label{eq:nominal system new}
\left\{
\begin{aligned}
  \dot x_{\nt}(t) &= A x_{\nt}(t) + B \bar u_{\text{opt}}(t), \quad x_{\nt}(0) = x_0, \\
  u_{\nt}(t) &= \bar u_{\text{opt}}(t),
\end{aligned}\right.
\end{equation}
where the transformed control input is defined as $\bar u_{\text{opt}}(t) = u_{\text{opt}}(t) + K_x x_{\nt}(t)$.  
By substituting $\bar u_{\text{opt}}(t)$ into the overall control law in \eqref{eq:total-control-law}, we obtain  
\begin{align}
\label{eq:total control law new}
    u(t) = \bar u_{\text{opt}}(t) + K_x(x(t)-x_{\nt}(t)) + u_a(t),
\end{align}
which has a similar structure to a \textit{tube-based MPC control law}~\cite{mayne2005robust}, except for the adaptive component $u_a$.  

{\bf Why use two nominal systems?} We first introduce the nominal system in its original form \eqref{eq:nominal system} because the performance analysis of \loneAC---specifically the derivation of its uniform performance bounds ---relies on comparing this nominal system with the actual closed-loop dynamics, where their difference arises solely from system uncertainties. These performance bounds will be discussed in detail in \cref{sec:l1acbounds}.
The transformation in \eqref{eq:nominal system new} reveals a {\it close connection between tube MPC and our proposed UC-MPC}. In both frameworks, the feedback term $K_x(x - x_{\nt})$ serves to keep the actual trajectory within a bounded neighborhood of the nominal one, while the adaptive component $u_a$ in UC-MPC actively compensates for matched uncertainties to further improve the performance. Owing to this adaptive compensation, the resulting deviation tube in UC-MPC is expected to be significantly tighter than that of the standard tube MPC.
This reformulation also allows the MPC input $u_{\text{opt}}$ in \eqref{eq:nominal system} to be computed directly from \eqref{eq:nominal system new}, which defines a {\it standard uncertainty-free MPC problem} as in \cite{rawlings2020mpc-book,chen1998quasi,grune2017nonlinear}. Consequently, the stability and recursive feasibility of the nominal MPC can be established using existing MPC analysis tools, which will be discussed in detail in \cref{sec:stability_and_recursive}. 
To summarize, the nominal system in \eqref{eq:nominal system} and control law in \eqref{eq:total-control-law} serve as the foundation for all performance bounds analysis in \cref{sec:l1acbounds}, while the equivalent transformed form in \eqref{eq:nominal system new} and control law in \eqref{eq:total control law new} are introduced for formulating the nominal MPC problem and highlighting the connection between UC-MPC and tube MPC~\cite{mayne2005robust}, as well as for the stability and recursive feasibility analysis of the nominal MPC problem within the UC-MPC framework.

\section{Unmatched Uncertainty Attenuation}
\label{sec:kx_design}
\vspace{-2mm}
An effective approach to reduce the constraint tightening on selected critical output channels is to design the feedback gain $K_x$ to minimize the influence of $w(t)$ on those channels. This can be achieved by minimizing the peak-to-peak gain (PPG) of the unmatched uncertainty on those channels, thereby limiting its adverse effects and allowing for less tightened constraints without compromising safety or performance. It is important to note that such an approach may reduce the tightening on critical channels at the cost of increasing the tightening on other channels.

Consider the following system:
\begin{subequations}\label{eq:kx_design-total}
    \begin{align}
        \dot x(t) &= Ax(t) + Bu(t) +B_uw(t), \; x(0) =x_0, \label{eq:kx_design} \\
          z(t) & = C_zx(t)+D_zu(t)   \label{eq:z}
    \end{align}
\end{subequations}where
$x(t)\in\mbR^n$ and $u(t)\in \mbR^m$ represent the state and input vectors, respectively, $w(t)\in\mbR^p$ represents the unmatched uncertainty, $z(t)$ denotes output channels whose bounds we aim to minimize, and matrices $A,B$, and $B_u$ are from \eqref{eq:dynamics-uncertain-original}. Note that when the matched uncertainty is perfectly eliminated, the original system \eqref{eq:dynamics-uncertain-original} is identical to the system in \eqref{eq:kx_design}. 
With the feedback control law $u(t) = K_xx(t)$, we achieve the following closed-loop system:
\begin{equation}\label{eq:kx_closed}
H\left\{
\begin{aligned}
 \dot x(t) &= (A+BK_x)x(t) +B_uw(t), \\ 
 z(t) & = (C_z+D_zK_x)x(t),\; x(0) =x_0, 
\end{aligned}\right. 
\end{equation}
where we use $H$ to represent the transfer function of this system with input $w(t)$ and output $z(t)$. Now, our goal is to minimize PPG of $H$, i.e., $\|H\|_{\text{peak}}$, by designing $K_x$, where $\|H\|_{\text{peak}} = \sup \{\|z(\tau)\|:x(0) =0, \tau\ge 0, \|w(t)\| \le 1 \text{ for } t\ge 0\}$.
The authors of \cite{scherer1997multiobjective} proposed to verify the upper bounds for $\|H\|_{\text{peak}}$ as follows.
\begin{lemma}(\cite[Section III.F]{scherer1997multiobjective}) \label{lemma:1}
The system in \eqref{eq:kx_closed} is internally stable and has a PPG bound of $\beta > 0$, if there exist $\lambda > 0$, $\mu > 0$, and a symmetric matrix $P$, s.t. the following two inequalities hold:
\begin{align}
    \begin{bmatrix}
        (A\!+\!BK_x)^\top P \!+\!P(A\!+\!BK_x)\!+\!\lambda P & PB_u \\ B^\top_uP & -\mu I
    \end{bmatrix} 
    < 0, \label{eq:theorem1_1}\\
    \begin{bmatrix}
        \lambda P & 0 & (C_z+D_zK_x)^\top \\ 0 & (\beta-\mu)I & 0 \\ C_z+D_zK_x & 0 & \beta I
    \end{bmatrix} 
    > 0. \label{eq:theorem1_2}
\end{align}
\end{lemma}
\vspace{-3mm}
From \cref{lemma:1}, we propose the following approach to design the feedback gain $K_x$ based on an LMI framework. It relies on the following modification of \cref{lemma:1}:
\begin{lemma} \label{lemma:ppg}
 The system in \eqref{eq:kx_design} is internally stable and the closed-loop system \eqref{eq:kx_closed} has a PPG bound of $\beta > 0$ with the controller $u = \mathcal{Y}\mathcal{V}^{-1}x$, if there exit a matrix $\mathcal{Y}\in \mbR^{m\times n}$, a positive definite matrix $\mathcal{V} \in \mbR^{n\times n}$, $\lambda > 0$, and $\mu > 0$, such that
    \begin{align}
        \begin{bmatrix}
           \left\langle A\mathcal{V}+B\mathcal{Y}
        \right\rangle + \lambda \mathcal{V} & B_u \\ B^\top_u & -\mu I
        \end{bmatrix} < 0, \label{eq:lemma1_1}\\
        \begin{bmatrix}
            \lambda \mathcal{V} & 0 & (C_z\mathcal{V}+D_z\mathcal{Y})^\top \\ 0 & (\beta-\mu) I & 0 \\ C_z\mathcal{V}+D_z\mathcal{Y} & 0 & \beta I  
        \end{bmatrix} > 0, \label{eq:lemma1_2}
    \end{align}
    where $\langle V\rangle$ is a shorthand notation for $V^\top+V$.
\end{lemma}
\vspace{-3mm}
\begin{proof}
Multiplying \eqref{eq:lemma1_1} by 
$\begin{bmatrix}
        \mathcal{V}^{-1} & 0\\ 0 & I
    \end{bmatrix}$ and its transpose from left and right gives \eqref{eq:theorem1_1}, where $P = \mathcal{V}^{-1}$ and $K = \mathcal{Y}\mathcal{V}^{-1}$. Similarly, multiplying \eqref{eq:lemma1_2} by 
    $\textup{diag}(\mathcal{V}^{-1},I,I)$
    and its transpose from left and right gives \eqref{eq:theorem1_2}. 
\end{proof}\vspace{-3mm}
Thus, we can leverage \cref{lemma:1} to achieve the internal stability and PPG bound $\beta$. Based on Lemma \ref{lemma:ppg}, we can derive matrices $\mathcal{Y}$ and $\mathcal{V}$ to minimize $\beta$ for a given $\lambda$ by solving the following LMI:
\begin{equation}
    \label{eq:ppg_op}
    \min_{\beta, \mu >0,\mathcal{Y},\mathcal{V}>0} \beta \;
    \text{s.t. \cref{eq:lemma1_1} and \cref{eq:lemma1_2} hold}. 
\end{equation}
By setting $K_x = \mathcal{Y}\mathcal{V}^{-1}$, we achieve the minimized PPG bound $\beta$ for system $H$ and less tightened constraint for channels in $z(t)$. Here, we have to pick a $\lambda$ first since \cref{eq:theorem1_1} and \cref{eq:theorem1_2} are only linear if we fix $\lambda$ \cite{scherer1997multiobjective}. For the selection of $\lambda$, a linear search between 0 to some predefined maximum value can be implemented for an improved performance.
% \vspace{-3mm}

\begin{remark}[Design of $K_x$]
\cref{lemma:ppg} provides one method to synthesize a state-feedback controller $u=K_x x $ to minimize the effect of bounded disturbances. There are other methods available in the literature, e.g., \cite{abedor1996linear} and \cite{khlebnikov2011optimization}, both of which leverage inescapable ellipsoids to over-approximate the reachable sets of the disturbed system \cref{eq:kx_design-total} and provide necessary and sufficient conditions for ensuring that system trajectories remain in prescribed ellipsoids. These methods can be seamlessly integrated into our proposed UC-MPC framework to replace the approach stated in \cref{lemma:ppg}, and can potentially lead to less conservative performance. 
\end{remark}
\vspace{-3mm}
\section{\loneAC~with Uniform Performance Bounds}\label{sec:l1acbounds}\vspace{-2mm}
In this section, we derive the uniform performance bounds on the errors between the states and control inputs of the real plant \cref{eq:dynamics-uncertain} with \loneAC~and those of the nominal system \cref{eq:nominal system}. A preliminary result is presented in \cite{zhao2023integrated}, in which only matched uncertainty is considered. Here, we extend the result to systems with both matched uncertainty and unmatched uncertainties, and provide two techniques to achieve meaningful performance bounds and less conservative results. In this section, we first present three key components of \loneAC~and the requirements of \loneAC~design in order to achieve vector-$\infty$ norm performance bounds. Building on these elements, we present a system transformation to achieve individual element bounds in \cref{sec:individual bounds}, and introduce a scaling method in \cref{sec:scaling of w} to obtain tighter results and reduce conservativeness in constraint tightening.

There are three key elements in an \loneAC:  a low-pass filter, a state predictor, and an estimation law.
The low-pass filter $\mcC(s)$ decouples the estimation loop from the control loop, enabling fast adaptation without sacrificing the robustness \cite{hovakimyan2010L1-book}. In practice, we can design the {\bf low-pass filter} $\mcC(s)$ as a first-order transfer function matrix
\begin{equation}\label{eq:filter-defn}
\mcC(s) = \textup{diag}(\mcC_1(s), \dots, \mcC_m(s)), \ \mcC_j(s) \trieq \frac{k_f^j}{(s+ k_f^j)},
\end{equation}
where  $k^j_f$ ($j\in\Zm$) represents the bandwidth for the $j$th input channel. 
Given any compact set $\mcX_0$, the stability of the closed-loop system can be guaranteed when there exists a positive constant $\rho_r$ and a (small) positive constant $\gamma_1$ such that
\begin{subequations}\label{eq:l1-stability-condition-w-extra-Lf}
\begin{align}
\lonenorm{\mcG_{xm} (s )} b_{f,\mcX_r}  & <  \rho_r   -  \lonenorm{\mcH_{xm}(s)} \linfnorm{u_{\text{opt}}} \nonumber \\&- \lonenorm{\mcH_{xu}(s)}b_w - \rho_{\textup{in}},\label{eq:l1-stability-condition} \\ 
\lonenorm{\mcG_{xm}(s)}L_{f,\mcX_a}&<1,\label{eq:l1-stability-condition-Lf}
\end{align}
\end{subequations}
where 
\begin{subequations}\label{eq:l1-norms-defns}
\begin{align} 
 \rho  & \trieq \rho_r + \gamma_1, \label{eq:rho-defn} \\
 \mcX_r & \trieq \Omega(\rho_r), \   \mcX_a  \trieq  \Omega(\rho). \label{eq:Xr-Xa-defn}\\
 \mcH_{xm}(s) & \trieq\! (sI_n \!-\! A_m)^{-1}\! B,\label{eq:Hxm_defn}\\
 \mcH_{xu}(s) &\trieq\! (sI_n \!-\! A_m)^{-1} \!{B}_u, \label{eq:Hxu-defn}\\
 \mcG_{xm}(s) & \trieq\! \mcH_{xm}(s)(I_m-\mcC(s)),
   \label{eq:Gm-defn} \\
 \rho_{\textup{in}} &\trieq \lonenorm{s(sI_n-A_m)^{-1}}\max_{x_0\in \mcX_0}\infnorm{x_0}. \label{eq:rho_in}
\end{align}
\end{subequations}

\begin{remark}
   Since $A_m$ is Hurwitz and $\mcC(s)$ is a strictly proper transfer function, all the \lonew norms of transfer functions in \cref{eq:l1-norms-defns} are bounded.
\end{remark}
\vspace{-2mm}

For the closed-loop system in \cref{eq:dynamics-uncertain}, we consider the {\bf state predictor} as 
\begin{multline}\label{eq:state-predictor}
\dot {\hat x} (t)\!=\! A_m x(t) \! +\! B u_{\text{opt}}(t) \!+\! B u_{\at} (t)\! +\! \hsigma(t) \!+\!  A_e \tilx(t),% + B^\perp \hsigma_2(t) \\
\end{multline}
where $\hat x(0) = x_0$, $\tilx(t) = \hat x(t) - x(t)$ is the prediction error, ${A_m} \trieq A + B{K_x}$ is Hurwitz, $A_e$ is a custom Hurwitz matrix, and $\hsigma (t)$ is the estimate of the lumped uncertainty, $Bf(t,x)+{B}_u{w}(t)$. We update the estimate $\hsigma(t)$ based on the following piecewise-constant {\bf estimation law} (similar to that in \cite[Section~3.3]{hovakimyan2010L1-book}):
\begin{equation}\label{eq:adaptive_law}\left\{
\begin{aligned}
&  \hsigma(t) &  &  \hspace{-4mm}=  
\hsigma(iT)
    , \quad t\in [iT, (i+1)T), \\
& \hsigma(iT) &
& \hspace{-4mm}= - \Phi^{-1}(T)e^{A_eT}\tilde{x}(iT),
\end{aligned}\right.
\end{equation} where $T$ is the estimation sampling time,  $i=0,1,2,\dots$, and $\Phi(T)\trieq A_e^{-1}\left(e^{A_eT}\!-I_n\right)$. The {\bf adaptive control law}  is defined as
\begin{equation}\label{eq:l1-control-law}
   u_{\at}(s) = -\mcC(s)\laplace{B^\dagger \hsigma(t)},
\end{equation}
where $u_a(s)$ is the Laplace transform of $u_a(t)$, $B^\dagger=(B^\top B)^{-1}B^\top$ is the pseudo-inverse of $B$.
As we explained in \cref{remark:B}, $B^\dagger \hsigma(t)$ represents the estimate of the matched uncertainty $f(t,x(t))$, and thus the adaptive control law \cref{eq:l1-control-law} is aimed at canceling the matched uncertainty only within the bandwidth of the filter $\mcC(s)$.
We define the following constants for future reference:\vspace{-2mm}
\begin{subequations}\label{eq:alpha_012-gamma0-T-defn}
\begin{align}
    \bar \alpha_0(T)  \trieq& \max_{t\in [0,T]} \int_0^t\infnorm{e^{A_e(t-\tau)}B}d\tau,\label{eq:alpha_0-defn}\\
    \bar \alpha_1(T)  \trieq& \max_{t\in [0,T]}\int_0^t\infnorm{e^{A_e(t-\tau)}{B}_u}d\tau,\label{eq:alpha_1-defn} \\
    \bar \alpha_2(T) \trieq& \max_{t\in [0,T]} \infnorm{e^{A_e t}},
   \label{eq:alpha_2-defn} \\
      \bar \alpha_3(T) \trieq& \max_{t\in [0,T]}\int_0^t \infnorm{e^{A_e (t-\tau)}\Phi^{-1}(T)e^{A_eT}}d\tau,\label{eq:alpha_3-defn}
         \\
        \gamma_0(T) \trieq& (b_{f,\mcX_a} \bar\alpha_0(T)\!+\!\bar \alpha_1(T)b_w)\!(\bar \alpha_2(T) \!+\! \bar \alpha_3(T)\!+\!1\!)\!\label{eq:gamma_0-T-defn} 
        \\ \hspace{-4mm}    \rho_{ur} \trieq&  \lonenorm{\mcC(s)} b_{f,\mcX_r}, \label{eq:rho_ur-defn} \\
\hspace{-2mm}  \gamma_2  \trieq &  \lonenorm{\mcC(s)}\! L_{f,\mcX_a} \gamma_1 \nonumber \\ &\!+\! \lonenorm{\mcC(s)B^\dagger(sI_n\!-\!A_e)}\!\gamma_0(T), \label{eq:gamma2-defn} \\
% \hspace{-4mm} \rho \trieq & \rho_r+\gamma_1 \label{eq:rho-defn}, \\ 
 \hspace{-4mm}  \rho_{u_{\at}} \trieq&  \rho_{ur} +\gamma_2, \label{eq:rho-u-defn}
\end{align}
\end{subequations}
where $\gamma_1$ is introduced in \cref{eq:rho-defn}. Based on the Taylor series expansion of $e^{A_eT}$, we have that 
\\
$\lim_{T\rightarrow 0} \int_0^T \infnorm{\Phi^{-1}(T)}d\tau$ is bounded, which further implies that $ \lim_{T\rightarrow 0} \bar \alpha_3(T)$ is bounded. 
Since $\lim_{T\rightarrow 0}\bar \alpha_0(T) =0 $, $\lim_{T\rightarrow 0}\bar \alpha_1(T) =0$, $\lim_{T\rightarrow 0} \bar \alpha_2(T)=0$, $b_{f,\mcX_a}$ is bounded for a compact $\mcX_a$, $b_w$ is bounded, and $\lim_{T\rightarrow 0} \bar \alpha_3(T)$ is bounded, we have
\begin{equation}\label{eq:gammaT-to-0}
    \lim_{T\rightarrow 0} \gamma_0(T) = 0.
\end{equation}
Considering \cref{eq:gammaT-to-0,eq:l1-stability-condition-Lf}, it is always feasible to find a small enough $T>0$ such that
\begin{equation}\label{eq:T-constraint}
\frac{\lonenorm{\mcH_{xm}(s)\mcC(s)B^\dagger(sI_n-A_e)}}{1-\lonenorm{\mcG_{xm}(s)}L_{f,\mcX_a}}\gamma_0(T) < \gamma_1,
\end{equation}
with $\mcX_a$ in \cref{eq:Xr-Xa-defn} and $B^\dagger$ being pseudo-inverse of $B$. 

\vspace{-3mm}
Following the typical analysis steps of \loneAC \cite{hovakimyan2010L1-book}, we introduce the following reference system for performance analysis: 
\begin{align}
  \dot{x}_\rt(t) &\! =\! {A_m}x_{\textup{r}}(t) \!\!+\!\! {B}u_{\text{opt}}(t) \!\!+\! \!B({u_{\textup{r}}}(t)\! \!+\!\! f(t,\!x_{\textup{r}}(t)))\!\!+\!\!{B}_u{w}(t), \hfill \nonumber   \\
  {u_{\textup{r}}}(s) & \! =\! -\mcC(s)\laplace{f(t,x_{\textup{r}}(t))}, \quad\xr(0)\! =\!x_0.\label{eq:ref-system} 
  %\\   y_{\textup{r}}(t) & \! =\! Cx_{\textup{r}}(t),, \hfill 
\end{align}
\begin{lemma}\label{them:x-xref-bnd}
{Given the uncertain system \cref{eq:dynamics-uncertain} subject to  \cref{assump:lipschitz-bnd-fi}, the reference system \cref{eq:ref-system} subject to the conditions \cref{eq:l1-stability-condition-Lf,eq:l1-stability-condition} with a constant $\gamma_1>0$, and an \loneAC~defined via \cref{eq:state-predictor,eq:adaptive_law,eq:l1-control-law}, subject to the sample time constraint  \cref{eq:T-constraint}, we have }
\begin{subequations}
\begin{align}
    \linfnorm{x} & \leq \rho, \label{eq:x-bnd}\\
    \linfnorm{u_{\at}} & \leq \rho_{u_{\at}}, \label{eq:ua-bnd}\\
   % \linfnorm{\tilde{x}} &\leq \bar{\gamma}_0, \label{eq:sigmatilde-xtilde-bnd}\\
    \linfnorm{x_{\textup{r}}-x} &\leq \gamma_1, \label{eq:xref-x-bnd}\\
    \linfnorm{u_{\textup{r}}-u_{\at}} &\leq \gamma_2, \label{eq:uref-u-bnd}
    % \\    \linfnorm{y_{\textup{r}}-y} &\leq\infnorm{C} \gamma_1,\label{eq:yref-y-bnd}
\end{align}
\end{subequations}
where $\rho$, $\rho_{u_{\at}}$, and $\gamma_2$ are defined in \cref{eq:rho-defn}, \cref{eq:rho-u-defn}, \cref{eq:gamma2-defn}, respectively. 
\end{lemma}

\begin{lemma} \label{lem:ref-id-bnd}
Given the reference system \cref{eq:ref-system} and the nominal system \cref{eq:nominal system}, subject to \cref{assump:lipschitz-bnd-fi}, and the condition  \cref{eq:l1-stability-condition}, we have
\begin{align}
{\left\| {{x_{\textup{r}}} - \xn} \right\|_{{\mathcal L}{_\infty }}} & \leq  \lonenorm{\mcG_{xm}}b_{f,\mcX_r} +\lonenorm{\mcH_{xu}}b_w \label{eq:xref-xid-bnd} .
\end{align}
\end{lemma}
\vspace{-4mm}
The proofs of \cref{them:x-xref-bnd} and \cref{lem:ref-id-bnd} can be obtained by extending the proofs of \cite[Theorem~1 and Lemma~5]{zhao2023integrated} and are included in the appendix. The reference system \eqref{eq:ref-system} requires knowledge of the true uncertainty and is therefore not implementable in practice. It is introduced purely for analysis purposes. It serves as an analytical bridge between the actual adaptive system and the nominal system, enabling explicit performance bounds to be derived.
\vspace{-3mm}
\subsection{Individual Performance Bounds}\label{sec:individual bounds}\vspace{-2mm}
The previous results provide uniform error bounds as represented by the vector-$\infty$ norm, which always leads to the {\it same bound} for all channels in the states, $x_i-x_{\nt,i}(t)$ ($i\in\mbZ_1^n$), or all channels in the adaptive inputs, $u_{\at,j}$ ($j\in\mbZ_1^m$). The use of vector-$\infty$ norms may lead to conservative bounds, making it impossible to satisfy the constraints \cref{eq:constraints} or leading to significantly tightened constraints for the MPC design. Thus, an individual bound for each $x_i(t)-x_{\textup r,i}(t)$ is preferred.
To derive the individual bounds for each $i\in\Zn$, we follow \cite{zhao2023integrated} to introduce the following coordinate transformation for the reference system \cref{eq:ref-system} and the nominal system \cref{eq:nominal system}:
\begin{equation}\label{eq:coordinate-trans}
\left \{
    \begin{aligned}
        \check x_{\rt} & = \Tau_x^ix_{\rt},\quad \check x_{\nt} = \Tau_x^i \xn, \\
         \check A_m^i &= \Tau_x^iA_m (\Tau_x^i)^{-1},\\
         \check B^i  & = \Tau_x^iB,   \quad \check B^i _u = \Tau_x^i{B}_u,
     %    \\ \check  C^i& \leftarrow  C(\Tau_x^i)^{-1},\\
    \end{aligned}
    \right.
\end{equation}
where $\Tau_x^i\!>\!0$ is a diagonal matrix that satisfies 
\begin{align}
    \Tau_x^i[i]&=1, \ 0<\Tau_x^i[k]\leq 1, \ \forall k\neq i, \label{eq:Tx-i-cts}
\end{align}
and $\Tau_x^i[k]$ is the $k$th diagonal element. 
With the transformation \cref{eq:coordinate-trans}, the reference system \cref{eq:ref-system} is transformed into 
\begin{align}
  \dot {\check x}_\rt(t) & \!=\! \check A_m^i\check x_{\rt}(t)\!\! +\!\!  \check  B^iu_{\text{opt}}(t) \!\!+\!\!  \check B^i (u_{\rt}(t) \!\!+\!\! \check f(t, \check x_{\rt}(t)\!)\!)\!+\!\check {B}_u^i{w}(t),  \nonumber\\ 
    {u_{\textup{r}}}(s) & \! =\! -\mcC(s)\laplace{\check f(t,\check x_{\textup{r}}(t))},  \ \check x(0) \!=\! \Tau_x^ix_0,\label{eq:ref-system-transformed}
   % \\ y(t) &\! = \!\check C \check x_{\rt}(t),\ \check x(0) \!=\! \Tau_x^ix_0,
\end{align}
where 
\begin{equation}\label{eq:f-checkf-relation}
    \check f(t,\check x_{\rt}(t)) =  f(t, x_{\rt}(t))) =  f(t, (\Tau_x^i)^{-1}\check x_{\rt}(t))).
\end{equation}Given a set $\mcZ$, define 
\begin{equation}\label{eq:check-Z-defn}
    \check \mcZ\trieq \{\check z\in \mbR^n: (\Tau_x^i)^{-1}\check z \in \mcZ\}.
\end{equation}
For the transformed reference system \cref{eq:ref-system-transformed}, we have
\begin{subequations}\label{eq:Hxm-Hxv-Gxm-check-defn}
\begin{align}
  \mcH_{\check xm}^i(s) &\trieq (sI_n \!-\! \check A_m^i)^{-1}\! \check B^i  = \Tau_x^i \mcH_{xm}(s), \label{eq:Hxm-check-defn} \\
  \mcH_{\check xu}^i(s) & \trieq (sI_n \!-\! \check A_m^i)^{-1} \! \check B^i _u = \Tau_x^i \mcH_{xu}(s), \label{eq:Hxv-check-defn} \\
  \mcG_{\check xm}^i(s) & \trieq \mcH_{\check xm}^i(s)(I_m-\mcC(s)) =  \Tau_x^i \mcG_{xm}(s), \label{eq:Gxm-check-defn}
\end{align}
\end{subequations}
where $\mcH_{xm},~\mcH_{xu},~\mcG_{xm}$ are defined in \cref{eq:Hxm_defn,eq:Hxu-defn,eq:Gm-defn}. 
By applying the transformation \cref{eq:coordinate-trans} to the nominal system \cref{eq:nominal system}, we obtain 
\begin{equation}\label{eq:nominal-cl-system-transformed}
\hspace{-2mm}
\left\{
\begin{aligned}
  \dot {\check x}_{\textup n}(t) & = \check A_m^i\xcheckn(t) + \check B^iu_{\text{opt}}(t),\ \xcheckn(0) \!=\! \Tau_x^ix_0,  \\ 
  \check y_{\text{n}}(t) & = \check C \xcheckn(t).
\end{aligned}\right. 
\end{equation}
Letting $\xcheckin(t)$ be the state of the system 
    $\dot{\check x}_{\textup{in}}(t) = \check A_m^i \xcheckin(t)$ with $\xcheckin(0) = \check x_{\nt}(0) = \Tau_x^i x_0$, 
we have $\xcheckin(s) \trieq (sI_n-\check A_m^i)^{-1} \xcheckin(0) = \Tau_x^i(sI_n- A_m)^{-1} x_0$. Define
\begin{equation}\label{eq:check-rhoin-defn}
   \check \rho_{\textup{in}}^i \trieq \lonenorm{s\Tau_x^i(sI_n- A_m)^{-1}}\max_{x_0\in \mcX_0}\infnorm{x_0}.
\end{equation}
Similar to \cref{eq:l1-stability-condition}, for the transformed reference system \cref{eq:ref-system-transformed}, given any positive constant $\gamma_1$, the low-pass filter design now needs to satisfy:
\begin{align}
\lonenorm{\mcG_{\check xm}^i(s)}b_{\check f,\check \mcX_r}  &< \check \rho_r^i - \lonenorm{\mcH_{\check x v}^i(s)}\linfnorm{u_{\text{opt}}} \nonumber\\&- \lonenorm{\mcH_{\check x u}^i(s)}b_w - \check \rho_{\textup{in}}^i ,\label{eq:l1-stability-condition-transformed} %
\end{align}
where 
$\mcX_r$ is defined in \cref{eq:Xr-Xa-defn} and $\check \mcX_r$ is defined according to \cref{eq:check-Z-defn}, and $\check \rho_r^i$ is a positive constant to be determined. 

\begin{lemma}\label{lem:refine_bnd_xi_w_Txi}
Consider the reference system \cref{eq:ref-system} 
subject to \cref{assump:lipschitz-bnd-fi}, the nominal system \cref{eq:nominal system}, the transformed reference system \cref{eq:ref-system-transformed} and transformed nominal system \cref{eq:nominal-cl-system-transformed}, obtained by applying \cref{eq:coordinate-trans} with any $\Tau_x^i$ satisfying \cref{eq:Tx-i-cts}. 
%,  to \cref{eq:ref-system} and \cref{eq:nominal-cl-system}, respectively. 
Suppose that \cref{eq:l1-stability-condition} holds with some constants $\rho_r$ and $\linfnorm{u_{\text{opt}}}$.
Then, %with $\check \rho_{\textup{in}}^i$ defined in \cref{eq:check-xn-bnd-w-rhon-defn},
there exists a constant $\check \rho_r^i\leq \rho_r$ such that \cref{eq:l1-stability-condition-transformed}
%for the transformed system \cref{eq:ref-system-transformed} 
holds with the same $\linfnorm{u_{\text{opt}}}$. Furthermore, $\forall t\geq 0$, 
\begin{align}
% {\abs{x_{\textup{in},i}(t)}} & \leq  \check \rho_{\textup{in}}^i, \  \forall t\geq0, \label{eq:xin-i-bnd-from-trans} \\
% \abs{x_{\textup{n},i}(t)} & \leq  \check \rho_{\nt}^i, \  \forall t\geq0, \label{eq:xn-i-bnd-from-trans} \\
\abs{x_{\rt,i}(t)} & \!\leq\!  \check \rho_r^i, \label{eq:xr-i-bnd-from-trans} \\
    % \left\| {\check x_{\textup{r}} - \check x_{\nt}} \right\|_{{\mathcal L}{_\infty }} & \leq  \lonenorm{\mcG_{ \check xm}(s)}b_{f,\Omega(\rho_r)},\  \forall t\geq 0, \label{eq:xref-xid-check-bnd}   \\
    %  \xr(t) & \in \mcX_r, \ \forall t\geq 0, \label{eq:xr-bnd-refined}\\
    \abs{x_{\textup{r,i}}(t)\!-\!x_{\nt,i}(t)} & \!\leq\! \lonenorm{\mcG_{ \check xm}(s)}\!b_{f,\mcX_r} \!+\! \lonenorm{\mcH_{ \check xu}} \!b_w, \label{eq:xri-xni-bnd-from-trans}
\end{align}
where we re-define \begin{equation}\label{eq:Xr-defn}
   \mcX_r\trieq \left\{z\in\mbR^n: \abs{z_i}\leq \check \rho_r^i, i\in\Zn\right\}. 
\end{equation}
\end{lemma}

\begin{lemma}\label{them:xi-uai-bnd}
Consider the uncertain system \cref{eq:dynamics-uncertain} subject to \cref{assump:lipschitz-bnd-fi}, the nominal system \cref{eq:nominal system}, and the \loneAC~defined via \cref{eq:state-predictor,eq:adaptive_law,eq:l1-control-law}, subject to the conditions \cref{eq:l1-stability-condition-Lf,eq:l1-stability-condition} with constants $\rho_r$ and $\gamma_1>0$ and the sample time constraint \cref{eq:T-constraint}. Suppose that for each $i\in \Zn$, %\cref{eq:check-xn-bnd-w-rhon-defn} holds with a constant $\check \rho_{\nt}^i$ for the transformed nominal system \cref{eq:nominal-cl-system-transformed} and 
\cref{eq:l1-stability-condition-transformed} holds with a constant $\check \rho_r^i$ for the transformed reference system \cref{eq:ref-system-transformed} obtained by applying \cref{eq:coordinate-trans}. Then, $\forall t\ge0$, we have 
\begin{subequations}\label{eq:xui-xuni-bnd-from-trans-w-tilX-tilU-defn}
\begin{align}
%\abs{x_{\textup{n},i}(t)} & \leq  \check \rho_{\nt}^i, \  \forall t\geq0, \label{eq:xn-i-bnd-from-trans-2} \\ 
%\abs{x_{i}(t)} \leq \rho^i, \ 
    % \abs{x_i(t)-x_{\nt,i}(t)} \leq \tilde \rho^i,\ & \forall t\geq 0,\ \forall i\in\Zn \label{eq:xi-xni-bnd-from-trans-w-tilX-defn}\\
\hspace{-2mm}     {x(t)-x_{\nt}(t)} \!\in\! \tilde {\mcX}\! &\trieq\! \left\{z\!\in\!\mbR^n\!: \!\abs{z_i}\!\leq\! \tilde \rho^i, \  i\in \Zn \right\}\!,\ \label{eq:xi-xni-bnd-from-trans-w-tilX-defn}\\
        %  \abs{u_{\at,j}(t)} \leq \rho_{u_{\at}}^j, \ &\forall t\geq 0,\ \forall j\in\Zm, \label{eq:ua-i-bnd-w-Ua-defn} \\
 \hspace{-2mm} {u_{\at}(t)} \!\in\! \mcU_{\at}  \!&\trieq\! \left\{z\!\in\!\mbR^m\!: \!\abs{z_j}\!\leq\! \rho_{u_{\at}}^j, \ \! j\!\in\! \Zm \right \}\!, \label{eq:ua-i-bnd-w-Ua-defn} \\ %\label{eq:Ua-defn}
            %  \abs{u_j(t) - u_{\nt,j}(t)}\leq \tilde\rho_u^j, \ & \forall t\geq 0, \  \forall j\in\Zm \label{eq:uj-unj-bnd-w-tilU-defn},
 \hspace{-2mm}        u(t) - u_{\nt}(t)\!\in\!  \tilde\mcU   \!&\trieq\! \left \{z \!\in\! \mbR^m\!: \!\abs{z_j}\!\leq\!  \tilde\rho_u^j,
%\!\infnorm{K_x[j,:]} \max_{i\in\Zn}\tilde\rho^i\!,
\  j\!\in\! \Zm \right \}\!, \label{eq:uj-unj-bnd-w-tilU-defn} % \label{eq:uj-unj-bnd-w-tilU-defn} 
\end{align}
\end{subequations}where 
\begin{subequations}\label{eq:rhoi-tilrhoi-y-defn}
\begin{align}
 \hspace{-3mm}   &  \tilde \rho^i\!\trieq\! \lonenorm{\mcG_{ \check xm}^i(s)}\!\!b_{f,\mcX_r}\!+\!\lonenorm{\mcH^i_{ \check xu}} \!\!b_w +\!\gamma_1,\!\label{eq:rhoi-tilrhoi-defn} 
      % \\ \tilde \rho_y^j  & \!\trieq\! \sum_{i=1}^n \abs{C[j,i]}\tilde\rho^i, \label{eq:tilrhoi-y-defn}%\infnorm{C[j,:]} \max_{i\in \mbI({C[j,:]})}\tilde\rho^i 
\\
\hspace{-3mm} & \rho_{u_{\at}}^j  \!\trieq\! \lonenorm{\mcC_j(s)} b_{f_j, \mcX_r}\! +\! \gamma_2,  \tilde\rho_u^j \!\trieq \!  \rho_{u_{\at}}^j \!+\!
\sum_{i=1}^n\abs{K_x[j,i]}\tilde \rho^i
\label{eq:tilrho-u-j-defn}
\end{align}
\end{subequations}
with $\mcX_r$ defined in \cref{eq:Xr-defn}, and $C[j,i]$ denoting the $(j,i)$ element of $C$.%, and $\mbI({C[j,:]})$ denoting set of indexes corresponding to nonzero elements in $C[j,:]$. 
\end{lemma}
%}
\cref{lem:refine_bnd_xi_w_Txi} and \cref{them:xi-uai-bnd} are the extensions of \cite[Lemma 6 and Theorem 3]{zhao2023integrated} to account for the unmatched uncertainty $w(t)$. The proofs can be found in the appendix.
\vspace{1mm}

\begin{remark} [Tuning of $T^i_x$]
Suppose we achieve the $\mathcal L_\infty-$norm bounds from \cref{them:x-xref-bnd} and \cref{lem:ref-id-bnd}. Then, we are guaranteed to achieve an individual bound from applying the transformation in \eqref{eq:coordinate-trans} with any $T^i_x$ satisfying the requirement \cref{eq:Tx-i-cts}. The key idea is that, for the \textit{target channel} $i$ whose bound we wish to refine, we progressively reduce the magnitudes of all other channels. This ensures that the vector $\infty$-norm bound is ultimately dominated by the target channel alone. Consequently, the final performance bound from the \loneAC~is determined solely by the $i$th channel. Once the magnitudes of the other channels have been reduced below that of the $i$th channel, further decreasing $T_x^i[k]$ (for $k\neq i$) cannot yield a tighter bound for the $i$th channel.
\end{remark}
\begin{remark}[Tuning of Performance Bounds]
\label{rem:x-xid-discussion-Cs-T-refined} 
\cref{them:xi-uai-bnd} introduces a technique to individually bound \( x_i(t) - x_{\nt,i}(t) \) for each \( i \in \Zn \) and \( u_j(t) - u_{\nt,j}(t) \) for each \( j \in \Zm \) through coordinate transformations. Furthermore, in the absence of unmatched uncertainty, by reducing \( T \) and increasing the bandwidth of the filter \( \mcC(s) \), the value of \( \tilde{\rho}^i \) (\( i \in \Zn \)) can be made arbitrarily small. This implies that {\it the adaptive system states can be made arbitrarily close to the nominal system states}. Additionally, the bounds on \( u_{\at,j}(t) \) and \( u_j(t) - u_{\nt,j}(t) \) can be brought arbitrarily close to the bounds on the true matched uncertainty \( f_j(t, x) \) for \( x \in \mcX_a \) and each \( j \in \Zm \). We notice that the low-pass filter in \loneAC~decides the tradeoff between performance and robustness. While reducing $T$ can lead to high-frequencies in the estimation loop, an appropriately selected bandwidth and order for the low-pass filter should prevent these high-frequencies to enter the system. Additional details on the role and systematic design of the filter can be found in \cite{hovakimyan2010L1-book}.
\end{remark}
\vspace{-3mm}

\subsection{Scaling for Improved Bounds}\label{sec:scaling of w}
\vspace{-2mm}
The previous results leverage the $\mathcal{L}_\infty$ bound of $w(t)$, i.e., $b_w$, to calculate the individual performance bounds according to \eqref{eq:rhoi-tilrhoi-y-defn}. It is possible that the bound $b_{w}$ is too large for some $w_i$, and thus may lead to conservative performance bounds calculated within \loneAC. To address this issue, we propose a scaling technique for unmatched uncertainty. Consider the following transformed system below
\begin{equation}\label{eq:dynamics-uncertain-new}
\left\{ \begin{aligned}
  \dot x(t) &= Ax(t) + B(u(t) +  f(t,x(t)))+\bar{B}_u \bar{w}(t), \hfill \\
  y(t) &= Cx(t), \ x(0) = x_0,\\ 
\end{aligned}\right.    
\end{equation}
where $\bar{B}_u = B_u \Lambda$ and $\bar{w}$ = $\Lambda^{-1} w$ for some positive diagonal matrix $\Lambda$ s.t. $\linfnorm {\bar{w}_i(t)} = b_w \; \forall i$. It follows that $B_uw = \bar{B}_u\bar{w}$, and thus the system \eqref{eq:dynamics-uncertain-new} is identical to the original system \eqref{eq:dynamics-uncertain-original}, and the assumption in \eqref{eq:bnd-w} still applies\ to $\bar{w}(t)$. 
We can easily achieve the same results in \cref{sec:individual bounds} for the transformed system \eqref{eq:dynamics-uncertain-new} but with smaller performance bounds. This is because during the calculation of the performance bounds in \loneAC, we only leverage $\mathcal{L}_\infty$ bound $b_w$. The additional information of individual bounds $b_{w_j}$ can help us design $\Lambda$, which makes the elements in $\bar{B}_u$ smaller than $B_u$ as $b_{w_j} \le b_w$, and thus leading to smaller performance bounds for the transformed system \eqref{eq:dynamics-uncertain-new} compared with the original system \eqref{eq:dynamics-uncertain-original}.
In \cref{sec:scaling_sim}, we will show that with the transformation of $\bar{B}_u$ and $\bar{w}$ in \eqref{eq:dynamics-uncertain-new}, which we denote as scaling of $w(t)$, we achieve a smaller performance bound for every state and control input channel provided by \loneAC~compared with the original system \eqref{eq:dynamics-uncertain-original}.
\vspace{-3mm}
\section{UC-MPC: Robust Adaptive MPC via Uncertainty Compensation}
\vspace{-2mm}
\begin{algorithm}
\caption{\loneAC~Design}\label{al:1}
\begin{algorithmic}[1]
\Require{Uncertain system \cref{eq:dynamics-uncertain-original} subject to \cref{assump:lipschitz-bnd-fi} with constraint sets $\mcX$ and $\mcU$, set for the initial condition $\mcX_0$, $A_e$ for state predictor \cref{eq:state-predictor}, initial low-pass filter $\mcC(s)$ and sampling time $T$ to define an \loneAC, $\gamma_1$, tol, and feedback gain $K_x$}
\State Calculate $\linfnorm{u_\text{opt}}$ based on $u_\nt(t) =K_x \xn(t) + u_\text{opt}(t)$ with constraints $x_\nt\in \mcX$ and $u_\nt\in \mcU$ using Pontryagin set difference
\While{\cref{eq:l1-stability-condition} with $\mcX_r = \Omega(\rho_r) \cap \mcX$ or \cref{eq:l1-stability-condition-Lf} with $\mcX_a = \Omega(\rho_r+\gamma_1) \cap \mcX$
does not hold with any $\rho_r$ given the input $\gamma_1$ and initial low-pass filter $\mcC(s)$}
\label{line:l1-stability-conditions-in-l1rg}
\State Increase the bandwidth of $\mcC(s)$
\EndWhile \Comment{$\rho_r$, $\mcX_r$ and $b_{f,\mcX_r}$ will be computed, and the bandwidth of the low pass filter is updated} \label{line:Xr-b_f-under-csts}
\State Set $b_{f,\mcX_r}^{old} = b_{f,\mcX_r}$
\For{$i=1,\dots,n$}\label{line:Ti}
\State Select $\Tau_x^i$ satisfying\!~\cref{eq:Tx-i-cts} and apply the transformation\!~\cref{eq:coordinate-trans}\label{line:Tx_i-selection} 
% \State Compute $\check \rho_{\textup{in}}^i$ according to \cref{eq:check-rhoin-defn}
\State Evaluate \cref{eq:Hxm-Hxv-Gxm-check-defn} and 
compute $\check \rho_{\textup{in}}^i$ according to \cref{eq:check-rhoin-defn}
%compute $\check \rho_{\textup{in}}^i$ according to \cref{eq:check-xn-bnd-w-rhon-defn}
\State Compute $\check \rho_{r}^i$ that satisfies \cref{eq:l1-stability-condition-transformed}
\State Set $\rho^i =  \check \rho_r^i+\gamma_1$, $\tilde \rho^i= \lonenorm{\mcG_{ \check xm}^i(s)}\!\!b_{f,\mcX_r}\!+\!\lonenorm{\mcH^i_{ \check xu}} \!\!b_w +\!\gamma_1$
\EndFor

\State Set $\mcX_r\! = \!\left\{\!z\!\in\!\mbR^n\!:\! \abs{z_i}\!\leq\! \check\rho_r^i\right\}\cap \mcX$ and update $b_{f,\mcX_r}$ \label{line:Xr-update-in-l1rg}% \Comment{Update $\mcX_r$ defined in step~\ref{line:Omega_r-defn}}
\If{$b_{f,\mcX_r}^{old} - b_{f,\mcX_r}>\textup{tol}$}
\State Set $b_{f,\mcX_r}^{old} = b_{f,\mcX_r}$ and go to step \ref{line:Ti}
\EndIf
\State Set $\mcX_a = \{ z\in\mbR^n: \abs{z_i}\leq \rho^i,\ i\in\Zn\}\cap \mcX$ \label{line:Xa-update-in-l1rg}
\While{constraint \cref{eq:T-constraint} does not hold with $\mcC(s)$ from step~\ref{line:Xr-b_f-under-csts} and $\mcX_a$ from step~\ref{line:Xa-update-in-l1rg},}
\State Decrease $T$ 
\EndWhile \Comment{$T$ is updated}
\label{line:T}
\State With $\mcC(s)$, $\mcX_r$, and $T$ from step \ref{line:T}, compute $\gamma_0(T)$ from \cref{eq:gamma_0-T-defn} and $\gamma_2$ from \cref{eq:gamma2-defn}.
\For{$j=1,\dots,m$}
\State Compute $\rho_{u_{\at}}^j$ and $\tilde \rho_u^j$ according to \cref{eq:tilrho-u-j-defn}
 \EndFor 
\State Set $\tilde {\mcX}$ and $\tilde {\mcU}$ with $\{\tilde \rho^i\}_{i\in\Zn}$ and $\{\tilde \rho_{u}^j\}_{j\in\Zm}$ %from steps \ref{line:derive-sep-state-bnds-in-l1rg} and~\ref{line:derive-sep-input-bnds-in-l1rg} 
via \cref{eq:xui-xuni-bnd-from-trans-w-tilX-tilU-defn}

% \State \textbf{Offline:}
\State Set 
\begin{equation}\label{eq:cst-tightening}
    \mcX_{\nt} \trieq \mcX \ominus \tilde {\mcX},\  \mcU_{\nt} \trieq \mcU \ominus \tilde {\mcU}
\end{equation}
\Ensure{Tightened bounds for the nominal system \cref{eq:nominal system} and an \loneAC~to compensate for uncertainties}
\end{algorithmic}
\end{algorithm}
To apply the proposed UC-MPC framework, we first follow \cref{sec:kx_design} to design the feedback gain $K_x$ to reduce the impact from the unmatched uncertainty. Next, we need to design an \loneAC~to achieve the uniform performance bounds and the tightened constraints for the nominal MPC. 
According to \cref{sec:l1acbounds}, to design an \loneAC, the following procedure can be utilized:
\vspace{-2mm}
\begin{itemize}
    \item Design a low-pass filter in \eqref{eq:filter-defn} that satisfies \eqref{eq:l1-stability-condition-w-extra-Lf};
    \item Apply the transformation\!~\cref{eq:coordinate-trans} with $\Tau_x^i$ satisfying\!~\cref{eq:Tx-i-cts} to achieve $\check \rho_{r}^i$ that satisfies \cref{eq:l1-stability-condition-transformed};
    % \item With $\check \rho_{r}^i$, calculate $\mcX_r\! = \!\left\{\!z\!\in\!\mbR^n\!:\! \abs{z_i}\!\leq\! \check\rho_r^i\right\}\cap \mcX$ and $\mcX_a = \{ z\in\mbR^n: \abs{z_i}\leq \rho^i,\ i\in\Zn\}\cap \mcX$, where $\rho^i =  \check \rho_r^i+\gamma_1$.
        \item Design the sampling time $T$ that satisfies \eqref{eq:T-constraint} with the low-pass filter and $\check \rho_{r}^i$ above.
\end{itemize}\vspace{-2mm}
The details are summarized in \cref{al:1}. 
It is important to note that in step~\ref{line:l1-stability-conditions-in-l1rg} of \cref{al:1} we also impose the constraint that both $\xr(t)$ and $x(t)$ remain within $\mcX$ for all $t \geq 0$. These constraints can help reduce the size of the uncertainty that needs to be addressed and substantially decrease the conservatism of the proposed approach. Finally, we formulate an MPC problem for the nominal system \cref{eq:nominal system new} using tightened constraint sets.

\subsection{Nominal MPC Formulation}
\vspace{-3mm}
Building on the tightened constraints from \cref{al:1}, we now present the design of MPC for the nominal system and the method for determining \( \bar{u}_{\text{opt}} \) from the total control law \eqref{eq:total control law new}. Following our analysis in \cref{sec:tube mpc}, at any time instant \( \tau \), given the state \( x_{\nt}(\tau) \), \( \bar{u}_{\text{opt}}(\tau) \) is defined by the solution to the following optimization problem:
\begin{subequations}\label{eq:mpc}
  \begin{align}
 \min_{u_z(\cdot)}\; &l_f(x_{z}(T_f))+\int_{0}^{T_f} l(x_{z}(t),u_z(t)) \,dt  \label{eq:mpc-cost}\\
  \hspace{-7mm}  \text{s.t.}~     
  & \dot x_{z}(t)  = A x_{z}(t) +  Bu_z(t), t\in[0,T_f],\label{eq:mpc-dynamics-cst}\\
  & x_{z}(0)=x_{\nt}(\tau), \\
    &x_{z}(t)\in  \mcX_{\nt}, t\in[0,T_f], \label{eq:mpc-state-cst} \\
    &x_{z}(T_f) \in \mcX_f,\label{eq:termianl_state_cst}\\
    &u_z(t)\in  \mcU_{\nt}, t\in[0,T_f], \label{eq:mpc-input-cst}
\end{align}  
\end{subequations}
where $T_f$ is the time horizon, $l_f$ denotes the custom terminal cost function, $l$ denotes the running cost function, and $\mcX_f$ is the custom terminal constraint set.
Finally, we can apply the control $\bar{u}_{\text{opt}}$, defined as $\bar{u}_{\text{opt}}(\tau+t)=u^*(t)$, to the system, for $0\le t \le t_\delta$, where $u^*(\cdot)$ is the solution to the optimization problem \cref{eq:mpc} given $x_{\nt}(\tau)$, and $t_\delta$ is a sufficiently small sampling time. It is worth noting that $\bar{u}_{\text{opt}}$ is derived solely from the nominal closed-loop dynamics, without relying on the actual system state affected by uncertainties.

\vspace{-2mm}
\subsection{Stability and Recursive Feasibility of Nominal MPC}
\label{sec:stability_and_recursive}
\vspace{-2mm}
\begin{assumption}\label{assumption: stability and feasibility}
    The nominal MPC problem \cref{eq:mpc} is recursively feasible and the nominal closed-loop system \eqref{eq:nominal system new} is stable.
\end{assumption}
\vspace{-2mm}
Because \eqref{eq:mpc} is a standard nominal MPC problem, classical results in \cite{rawlings2020mpc-book,chen1998quasi,grune2017nonlinear} can be applied to ensure or verify recursive feasibility and closed-loop stability.
See \cite[Lemma 2 and Theorem 1]{chen1998quasi} for the recursive feasibility and stability analysis for continuous-time MPC using terminal constraints. The high level idea is to find a feedback control law $K_fx$, a terminal constraint set $\Omega_\alpha$, and a cost function design such that the terminal cost bounds the infinite horizon cost of the system starting from the terminal constraint set controlled by the linear state feedback.
\vspace{-2mm}
\subsection{Non-emptiness of the Tightened Constraint Sets}
\label{sec:nonemptiness}
\vspace{-2mm}
\begin{assumption}
    \label{assumption: nonemptiness}
    The tightened constraint sets $\mcX_{\nt}$ and $\mcU_{\nt}$ are nonempty.
\end{assumption}
\vspace{-2mm}
According to \cref{them:xi-uai-bnd} and as explained in \cref{rem:x-xid-discussion-Cs-T-refined}, in the absence of the unmatched uncertainty, by decreasing the sampling time  $T$ and increasing the bandwidth of the filter $\mcC(s)$, one can make $\mcX_{\nt}$ arbitrarily close to $\mcX$, and make the bound on $u_j(t)-u_{\nt,j}(t)$ (that defines $\tilde {\mcU}$) arbitrarily close to the bound on the true matched uncertainty $f_j(t,x)$ for each $j\in\Zm$. 
As a result, in the absence of unmatched uncertainties, with a sufficiently small $T$ and a sufficiently high bandwidth for $\mcC(s)$, nonempty $\mcX_{\nt}$ can always be achieved as long as $\mcX$ is nonempty, and nonempty $\mcU_{\nt}$ can be achieved if the set $\mcU$ is sufficiently large to allow for complete cancellation of the matched uncertainty. 

In the presence of unmatched uncertainties, we need $\mcX$ to be sufficiently large to additionally accommodate for the effect of unmatched uncertainties that cannot be directly canceled. To guarantee non-emptiness of $\mcU_{\nt}$, it is therefore necessary to assume that $\mcU$ is sufficiently large to cancel the matched component as well as to absorb the residual error from unmatched uncertainties amplified by $K_x$ according to \eqref{eq:tilrho-u-j-defn}.
This is analogous to constraint tightening in robust MPC formulations with pre-stabilizing feedback laws, such as tube MPC in \cite{mayne2005robust, langson2004robust-mpc-tube}, where the input constraint set is tightened as $\mcU_{\nt} = \mcU \ominus K_x Z$, with $Z$ being a disturbance-invariant set whose size depends on both system dynamics and the uncertainty bounds. In these formulations, feasibility critically depends on the joint design of $K_x$ and $Z$.

\begin{remark}
In the UC-MPC framework, $K_x$ is designed to mitigate the effect of unmatched uncertainties on critical output channels by minimizing the peak-to-peak gain. Unlike tube MPC formulation in \cite{limon2008design}, we do not employ an LMI-based joint design of $K_x$ and tightened constraint sets to ensure the non-emptiness of the tightened constraint sets for the following reason. In tube MPC, the feedback gain is optimized with respect to the total disturbance effect, while in UC-MPC, the adaptive component already compensates for the matched part of the uncertainty. Thus, $K_x$ is tuned solely to minimize the effect of the unmatched part of the uncertainty on particular output channels.
Incorporating the non-emptiness conditions of the tightened sets $\mcX_{\nt} \trieq \mcX \ominus \tilde {\mcX},$ and  $\mcU _{\nt} \trieq \mcU \ominus \tilde {\mcU}$ into an LMI-based design for UC-MPC would require explicit characterization of the residual error from partially canceled matched uncertainties. However, this residual bound depends on $K_x$ itself, making it intractable to include within the same convex synthesis.
\end{remark}

\subsection{Formal Guarantees of UC-MPC}
\vspace{-2mm}
\begin{theorem}\label{theorem:1}
Consider an uncertain system \eqref{eq:dynamics-uncertain-original} subject to \cref{assump:lipschitz-bnd-fi} with the state and input constraints in \eqref{eq:constraints}. Consider the control law \eqref{eq:total control law new}, where $u_a$ is the control input of an adaptive controller designed following \cref{al:1}, and $\bar u_{\text{opt}}$ is obtained from solving the nominal MPC problem \cref{eq:mpc}. If  \cref{assumption: stability and feasibility,assumption: nonemptiness} hold, we have:
\begin{align} \label{eq:final-theorem}
     x(t) \in \mathcal{X},\; u(t) \in \mathcal{U}, \; x(t) - \xn(t) \in \tilde {\mcX}, \quad \forall t\geq 0,
\end{align}
where $\xn(t)$ is the state of the nominal system \eqref{eq:nominal system new}, and $\tilde {\mcX}$ is defined in \eqref{eq:xi-xni-bnd-from-trans-w-tilX-defn} and \eqref{eq:rhoi-tilrhoi-defn}.
\end{theorem}
\vspace{-6mm}
\begin{proof}
Based on \cref{assumption: stability and feasibility} and \cref{assumption: nonemptiness}, we have bounded nominal system states and control inputs, and thus $\linfnorm{\bar{u}_{\text{opt}}}$ and $\linfnorm{{u}_{\text{opt}}}$ are bounded, which ensures that we can always find positive constant $\rho_r$ and $\gamma_1$ s.t. \eqref{eq:l1-stability-condition-w-extra-Lf} holds. Thus, \cref{them:xi-uai-bnd} can be leveraged with the formulation of MPC in \eqref{eq:mpc}, control law \eqref{eq:total-control-law} (equivalent to control law \eqref{eq:total control law new}), and the definition of $\mcX_{\nt}$ and $\mcU_{\nt}$ in \eqref{eq:Xn-Un-defn} to achieve \eqref{eq:final-theorem}. 
\end{proof}

\begin{remark}
According to \cref{theorem:1}, for an LTI system subject to matched nonlinear time-varying uncertainties and unmatched bounded uncertainties characterized by \cref{eq:dynamics-uncertain-original}, by following the proposed UC-MPC framework, we can not only enforce the constraints for the actual system, but also achieve uniform performance bounds for this actual system in terms of its state deviation from a  nominal uncertainty-free system. Additionally, from \cref{rem:x-xid-discussion-Cs-T-refined} and \cref{sec:nonemptiness}, we can make the states of the actual closed-loop system  arbitrarily close to those of the nominal system, i.e., making $\tilde X$ arbitrarily small, in the absence of unmatched uncertainties. This is in contrast to existing robust MPC schemes, which typically focus on constraint enforcement and stability only (without guaranteeing the control performance), as well as adaptive MPC schemes, \cite{adetola2011robust,sasfi2023robust,lorenzen2017adaptive}, which typically consider only parametric uncertainties and yield limited performance when the parameter estimation is poor.
\end{remark}
\vspace{-2mm}
\subsection{Computational Complexity Analysis}
Compared with tube MPC formulations in \cite{mayne2005robust, langson2004robust-mpc-tube}, the proposed UC-MPC framework significantly reduces the online computational burden. In tube MPC, the MPC optimization must be solved online at every sampling instant using the actual system state, and the nominal initial state is typically introduced as a decision variable constrained to be consistent with the actual state through the tube set. This increases the size and complexity of the optimization problem.
In contrast, the MPC problem in \eqref{eq:mpc} is solved online using only the nominal dynamics and nominal system state, without introducing the nominal initial state as a decision variable tied to the actual state. This simplification is enabled by the integration of \loneAC, which guarantees that the actual closed-loop system remains close to the nominal dynamics in~\eqref{eq:nominal system new}. As a result, the optimization problem remains comparable in complexity to a vanilla MPC while yielding improved performance and robustness against model uncertainties.

The UC-MPC framework introduces additional computation only in two components: \textbf{constraint tightening} and \textbf{adaptive compensation}.  
\vspace{-2mm}
\begin{itemize}
    \item \textbf{Constraint tightening:} The tightened state and input sets are computed {\it offline} from uniform error bounds determined by the design of the \loneAC~and the feedback gain \(K_x\). This involves solving a small set of LMI to obtain \(K_x\), followed by evaluating closed-form expressions for the corresponding error bounds. This offline procedure is comparable in complexity to that of tube MPC.

    \item \textbf{Adaptive compensation:} During online execution, the only additional computation arises from evaluating the adaptive compensation term \(u_a\), which requires propagating the state predictor~\eqref{eq:state-predictor} and updating the uncertainty estimate through the adaptive law~\eqref{eq:adaptive_law}. These updates are computationally much cheaper compared with the effort required to solve the MPC optimization problem.
\end{itemize}

Overall, UC-MPC preserves the robustness guarantees of tube-based approaches while significantly reducing real-time computational effort.

\section{Simulation Case Studies}
\vspace{-2mm}
In this section, we evaluate the performance of the proposed UC-MPC through two case studies: controlling the longitudinal motion of an F-16 aircraft and ensuring soft landing of a spacecraft on an asteroid. In the F-16 aircraft case, we address both matched and unmatched artificial uncertainties to enhance tracking performance compared to vanilla MPC and tube MPC. 
Moreover, leveraging the methodology outlined in Section \ref{sec:kx_design}, we design the feedback gain $K_x$ to minimize constraint tightening for safety-critical state channels, outperforming alternative $K_x$ designs.
% For the asteroid landing case, we focus on the matched uncertainty arising from the asteroid's uncertain gravity model. 
All simulations are performed in MATLAB using dynamic propagation, and the optimization problem defined in \eqref{eq:mpc} is solved using MATLAB toolbox YALMIP with solver MOSEK. MATLAB code is available at https://github.com/RonaldTao/UC-MPC.
\vspace{-3mm}
\subsection{F-16 Longitudinal Motion Control}
\vspace{-3mm}
\subsubsection{System Dynamics}
\vspace{-2mm}
The F-16 aircraft model is adopted from \cite{sobel1985design-pitch}, and it has been slightly simplified by neglecting the actuator dynamics. The open-loop dynamics are given by 
% { \setlength{\arraycolsep}{2pt}
% \begin{align}
%   \hspace{-1mm}  &\dot x = \begin{bmatrix}
%     0 & 0.0067& 1.34 \\
%     0 & -0.869 & 43.2 \\
%     0 & 0.993 & -1.34
%     \end{bmatrix} x \nonumber\\
%     &\!+\!\begin{bmatrix}
%     0.169 & 0.252 \\
%     -17.3& -1.58 \\
%     -0.169 & -0.252 
%     \end{bmatrix}\! (u\!+\!f(t,x)) 
%     \!+\!\begin{bmatrix}
%         0.1061 \\0 \\ 0.1061
%     \end{bmatrix} \!w(t), \label{eq:ol-dynamics-f16}
% \end{align}
\begin{align}
  \hspace{-1mm}  &\dot x = A x 
    +B (u+f(t,x)) 
    + B_u w(t), \label{eq:ol-dynamics-f16}
\end{align}
where 
{\setlength{\arraycolsep}{1pt}
\[ A \!= \!\!
\begin{bmatrix}
    0 & 0.0067& 1.34 \\
    0 & -0.869 & 43.2 \\
    0 & 0.993 & -1.34
    \end{bmatrix}\!,  
    B\! \mid \! B_u \!=\!
\left[
\begin{array}{cc|c}
0.169  & 0.252  & 0.1061 \\
-17.3  & -1.58  & 0 \\
-0.169 & -0.252 & 0.1061
\end{array}\right]\!,
    % B \!=\!\! \begin{bmatrix}
    % 0.169 & 0.252 \\
    % -17.3& -1.58 \\
    % -0.169 & -0.252 
    % \end{bmatrix}\!, B_u\! = \!\!\begin{bmatrix}
    %     0.1061 \\0 \\ 0.1061
    % \end{bmatrix},
\]}
where the state $x(t) = [\gamma(t), q(t), \alpha(t)]^\top$ consists of the flight path angle, pitch rate, and angle of attack, the control input $u(t)=[\delta_e(t), \delta_f(t)]$ consists of the elevator deflection and flaperon deflection, and
\begin{align}
    \label{eq:f}
    f(t,x) &= [-1.44\sin(0.4\pi t)-0.18\alpha^2, 0.18-0.36\alpha]^\top, \nonumber\\
    w(t)&=3\sin(0.6\pi t)
\end{align}
are assumed uncertainties that depend on both time and $\alpha$. 
The output vector is defined as \( y(t) = [\theta(t), \gamma(t)]^\top \), where \( \theta(t) = \gamma(t) + \alpha(t) \) represents the pitch angle. The objective is to ensure that the output vector \( y(t) \) tracks the reference trajectory \( r(t) = [\theta_c(t), \gamma_c(t)]^\top \), where \( \theta_c(t) \) and \( \gamma_c(t) \) denote the desired pitch angle and flight path angle, respectively. The system is subject to the following state and control constraints:
\begin{equation}\label{eq:cts-F16}
    \abs{\alpha(t)} \leq 4 \!\textup{ deg},\ \ \abs{\delta_e(t)} \leq 25 \!\textup{ deg},\ \ \abs{\delta_f(t)} \leq 22 \!\textup{ deg}.
\end{equation}
Additionally, we assume:
\begin{equation}\label{eq:F16-init}
    \linfnorm{r} \leq 10, \quad x(0) \in \mcX_0 = \Omega(0.1).
\end{equation}
Our simulations indicate that under the given constraints and uncertainty formulation, \( |\gamma(t)| \leq 10 \) and \( |q(t)| \leq 100 \) for all $t$. Consequently, using the convention in \cref{eq:constraints}, the state constraints can be expressed as \( x(t) \in \mcX \triangleq [-10, 10] \times [-100, 100] \times [-4, 4] \).
\vspace{-3mm}
\subsubsection{UC-MPC Design}\label{sec:sub-l1rg-f16}
\vspace{-2mm}
Based on the uncertainty formulation in \cref{eq:f}, for any set \( \mcZ \in \mathbb{R}^3 \) (domain for system state $x(t)$), the values \( L_{f_1,\mcZ} = 0.36 \max_{\alpha \in \mcZ_3} \abs{\alpha} \), \( L_{f_2,\mcZ} = 0.36 \), \( b_{f_1,\mcZ} = 1.44 + 0.18 \max_{\alpha \in \mcZ_3} \alpha^2 \) and \( b_{f_2,\mcZ} = 0.18 + 0.36 \max_{\alpha \in \mcZ_3} \abs{\alpha} \) satisfy \cref{assump:lipschitz-bnd-fi}. Additionally, \( b_w = 3 \) satisfies \cref{assump:lipschitz-bnd-fi}. 
For the design of the \loneAC~controller in \cref{eq:l1-control-law,eq:adaptive_law,eq:state-predictor}, we chose \( A_e = -10I_3 \) and parameterized the filter as \( \mcC(s) = \frac{k_f}{s + k_f} I_2 \), setting the bandwidth for both input channels as \( k_f = 200 \). We also set \( \gamma_1 = 0.01 \) and the estimation sample time \( T \) to \( 10^{-7} \) seconds, which satisfies \cref{eq:T-constraint}. To achieve individual performance bounds, \( \Tau_x^i[k] = 0.01 \) was assigned for all \( i, k \in \mathbb{Z}_1^3 \) with \( k \neq i \), ensuring compliance with \cref{eq:Tx-i-cts}. 
The reference command \( r(t) \) was specified as \( [9, 6.5] \) deg for \( t \in (0, 7.5) \) seconds and \( [0, 0] \) deg for \( t \in (7.5, 15) \) seconds. For the MPC design for the nominal system, the time horizon was set to 0.2 seconds, and the cost function in \cref{eq:mpc} at time \( \tau \) was defined as:
% \begin{align*}
    $\mathbf{J}(x_{\nt}(\tau), u(\cdot)) = \int_{0}^{0.2} 100\|Cx(t) - r(\tau + t)\| 
    % \\
    % &\quad     
    + \|u(t)\| + 100\|\dot{u}(t)\| \, dt$.
% \end{align*}
The term \( \|\dot{u}(t)\| \) was included to penalize control inputs with large time derivatives, ensuring a smoother control input trajectory. For simplicity, we did not include terminal constraint in the MPC formulation.

\vspace{-3mm}
\subsubsection{Performance bounds with different $K_x$}
\vspace{-2mm}
For the longitudinal dynamics of the aircraft, the angle of attack is critical to safety, and thus we aim to design a feedback gain $K_x$ such that the effect of the unmatched uncertainty $w(t)$ on $\alpha(t)$ is as small as possible.
By solving \eqref{eq:ppg_op} with $C_z = [0,0,1;0,0,0;0,0,0]$ and $D_z = [0,0;0.1,0;0,0.1]$, the feedback gain of the feedback controller was obtained as $K_x = [4.1075,1.6612,16.2228;-3.0950,0.0595,-1.6320]$ with $A_m=A+BK_x$ being Hurwitz. Here, we specifically designed $D_z$ to include small weights on the control input to ensure that $K_x$ is not excessively large and thus it results in a reasonable control input.
%Based on the $f(t,x) \in \mathcal W=[-2.4,2.4]\times [-0.9,0.9]$ when $x\in\mcX$ holds. 
The calculated individual bounds following \cref{al:1} are listed as follows: $\tilde{\rho}=[0.40,1.28,0.11]$, $\tilde{\rho}_u=[7.40,2.26]$, and $\rho_{u_a}=[4.34;1.65]$. Thus, the constrained domains for the nominal system were designed as $x(t)\in\mcX_{\nt}\trieq [-9.60,9.60]\times [-98.72,98.72]\times[-3.89,3.89]$ and $u(t)\in\mcU_{\nt}\trieq [-15.07,15.07]\times [-18.86,18.86]$. To illustrate the effect of the design of $K_x$ on the constraint tightening parameter $\beta$, we selected other $K_x$ matrices that make $A_m$ Hurwitz and calculated the corresponding performance bounds. The results are presented in \cref{tab:Kx}. It shows that when $K_x$ leads to a smaller PPG bound $\beta$, $\tilde \rho^3$ will be smaller, indicating a smaller effect of $w(t)$ on $\alpha(t)$. However, we achieve the decrease in the state at the cost of a larger control bound, $\tilde \rho^1_u$.
\begin{table*}[t] % 't' aligns the table at the top of the page
\centering 
\caption{\loneAC~performance bounds and PPG bounds with different $K_x$} \label{tab:Kx}
\begin{tabular}{cllllll}
\toprule
$K_x$ & $\beta$ & $\tilde \rho^1$ & $\tilde \rho^2$ & $\tilde \rho^3$ & $\tilde \rho_u^1$ & $\tilde \rho_u^2$  \\\midrule
{[}3.25 0.89 7.12; -6.10 -0.90 -10.00{]} & 0.07  & 0.38                                   & 1.36                                                         & 0.27                                                         & 6.55                                                           & 4.64                                                           \\
{[}3.29 1.23 9.71; -2.75 -0.10 -2.67{]}  & 0.05  & 0.43                                   & 1.30                                                          & 0.17                                                         & 9.00                                                           & 3.41                                                           \\
{[}4.11 1.66 16.22; -3.10 0.06 -1.63{]}   & 0.03  & 0.40                                   & 1.28                                                         & 0.11                                                        & 9.93                                                           & 3.14 \\
\bottomrule
\end{tabular}
\end{table*}
\vspace{-3mm}
\subsubsection{Simulation Results}
\vspace{-2mm}
For comparison, we also implemented a vanilla MPC and a tube MPC (TMPC) \cite{mayne2005robust}. In the case of UC-MPC, the control law follows \cref{eq:total-control-law}: \( u(t) = K_xx(t) + u_{\text{opt}}(t) + u_a(t) \), where \( u_{\text{opt}}(t) \) is obtained by solving the optimization problem in \cref{eq:mpc} with the constrained domains \( \mcX_{\nt} \) and \( \mcU_{\nt} \) calculated earlier. For the implementation of the \loneAC, instead of using a sample time \( 10^{-7} \) used in \cref{sec:sub-l1rg-f16} for deriving the theoretical bounds, we used a larger sample time of \( 10^{-4} \) which does not satisfy \eqref{eq:T-constraint}. We demonstrate that the theoretical bounds derived under a sample time of \( 10^{-7} \) seconds still hold in simulations with a sample time of \( 10^{-4} \) seconds. These observations suggest that the calculated performance bounds are not limited to the sampling time used for theoretical derivations and it's not necessary to estimate the uncertainty at such a rate.
For the vanilla MPC method, the control input \( u(t) = u_{\text{opt}}(t) \) is determined using the nominal system at each time step. However, the optimization for the vanilla MPC becomes infeasible during the simulation because the nominal system is used to compute the control input without accounting for uncertainties. To address this issue and complete the simulation, state constraints were incorporated as soft constraints, which revert to the original constraints when feasible. 
For TMPC, we adopted the control law \( u(t) = u_{\text{opt}}(t) + K_x(x - x_{\nt}) \) as described in \cite{mayne2005robust}, where \( u_{\text{opt}}(t) \) is obtained by solving \cref{eq:mpc} with the constrained domains \( x \in \mathcal{X} \ominus Z \) and \( u \in \mathcal{U} \ominus K_xZ \). Here, \( Z \) is the disturbance invariance set for the controlled uncertain system \( \dot{x}(t) = (A + BK_x)x(t) + Bf(t, x(t)) + B_uw(t) \) with respect to the uncertainties \( f \) and \( w \). 
A discrete-time formulation was used for implementing all MPC methods.
\begin{figure}[t]
    \centering
    \includegraphics[width=1\columnwidth]{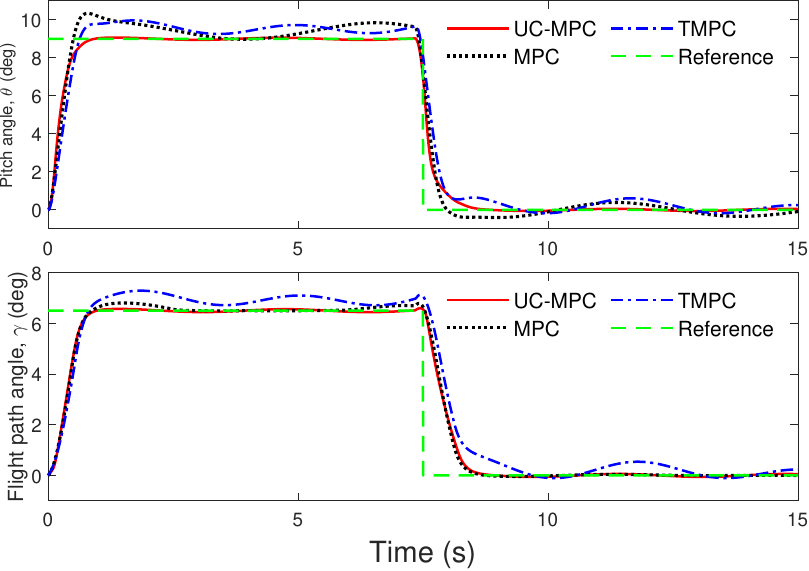} 
    \vspace{-6mm}
    \caption{Tracking performance under MPC, TMPC and UC-MPC (ours).}\label{fig:tracking}
        \vspace{-1mm}
\end{figure}
\begin{figure}[t]
    \centering
    \includegraphics[width=1\columnwidth]{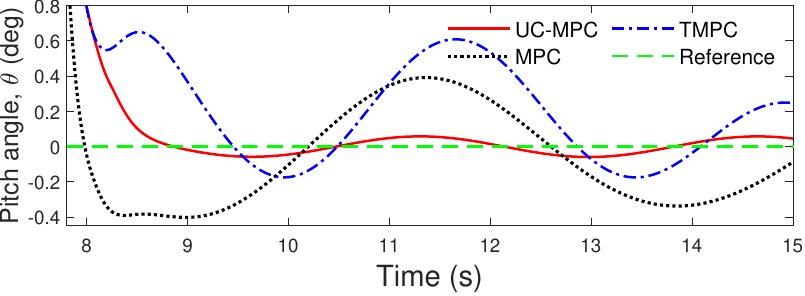} 
    \vspace{-6mm}
    \caption{Zoomed-in view of tracking performance on $\theta$ under MPC, TMPC and UC-MPC.}\label{fig:tracking_zoom}
        \vspace{-1mm}
\end{figure}
\begin{figure}[!t]
    \centering
    \includegraphics[width=1\columnwidth]{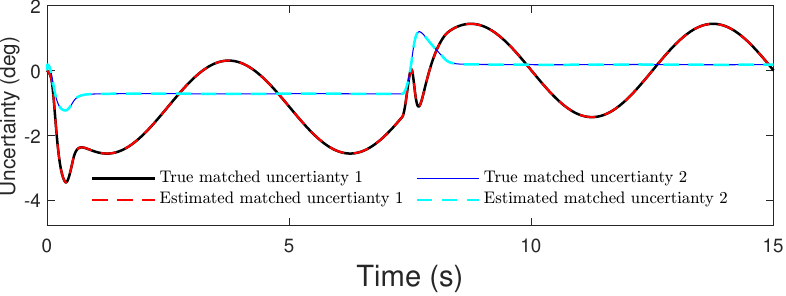} 
        \vspace{-6mm}
    \caption{Actual and estimated matched uncertainties $f(t,x(t))$ under UC-MPC. $f(t,x(t))$ has two elements, corresponding to 1 and 2 in figure.} 
    \label{fig:uncert}
        \vspace{-1mm}
\end{figure}
\begin{figure}[!t]
    \centering
    \includegraphics[width=1\columnwidth]{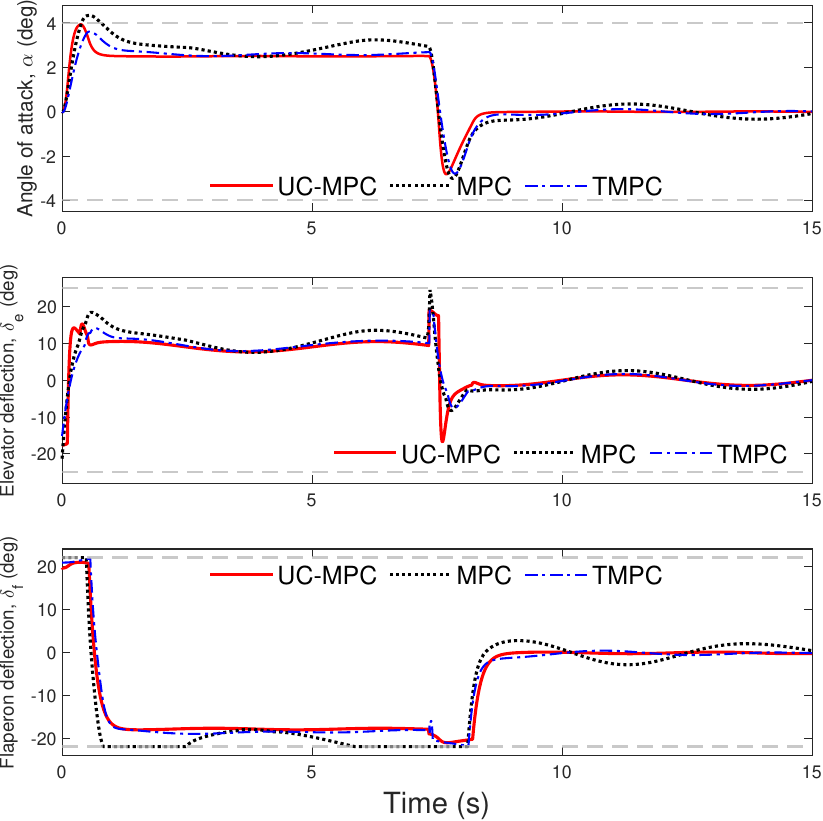}
    \vspace{-6mm}
    \caption{Trajectories of constrained states (top), and control inputs (middle and bottom) under MPC, TMPC, and UC-MPC. Gray dash-dotted lines illustrate the constraints specified in \cref{eq:cts-F16}.}    \label{fig:constraints}
        \vspace{-1mm}
\end{figure}
\begin{figure}[t]
    \centering
    \includegraphics[width=1\columnwidth]{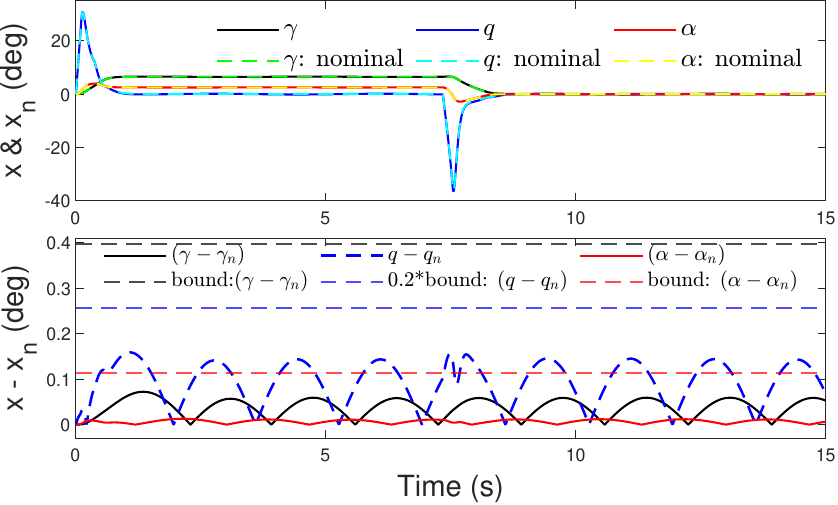} 
    \vspace{-6mm}
    \caption{Trajectories of states of the uncertain system ($x(t)$) under UC-MPC and of the nominal system ($x_{\nt}(t)$) and their differences. }    \label{fig:x-xn}
        \vspace{-1mm}
\end{figure}

\begin{figure}[t]
    \centering
    \includegraphics[width=1\columnwidth]{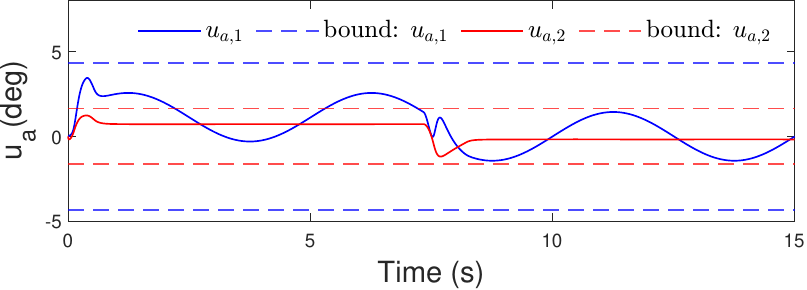} 
    \vspace{-6mm}
    \caption{Adaptive control inputs and theoretical bounds.}    \label{fig:ul1}
    \vspace{-1mm}
\end{figure}
\vspace{-3mm}
The simulation results are shown in \cref{fig:tracking,fig:tracking_zoom,fig:tracking_zoom_no_w,fig:constraints,fig:uncert,fig:x-xn,fig:ul1}.
First, \cref{fig:tracking,fig:tracking_zoom} show that UC-MPC resulted in an improved tracking performance compared with both MPC and TMPC, achieving a small error between the actual output $y(t)$ and the reference trajectory $r(t)$. Based on \cref{fig:uncert}, the matched uncertainty estimation from UC-MPC was accurate, and thus the proposed UC-MPC is able to approximately cancel the matched uncertainty $f(t,x(t))$ with the adaptive control input $u_a(t)$. According to the zoomed-in view in \cref{fig:tracking_zoom}, the output trajectories from UC-MPC did not precisely follow the reference trajectories in steady-state. The existence of the unmatched uncertainty $w(t)$ results in the oscillations of the closed-loop trajectory in steady-state, at the same frequency. In contrast, the output trajectories under MPC and TMPC were influenced by both the matched uncertainty and unmatched uncertainty, resulting in a larger tracking error in steady-state.

In terms of enforcing constraints, \cref{fig:constraints} demonstrates that both TMPC and UC-MPC managed to satisfy all constraints due to the implementation of constraint tightening. On the other hand, the MPC output breached the state constraint boundary as it did not account for uncertainties during optimization without the use of constraint tightening.
It is also worth mentioning that TMPC enforced the constraints at the cost of larger steady-state tracking error, as shown in \cref{fig:tracking_zoom}. 
Figure~\ref{fig:x-xn} displays the trajectories of the actual state $x(t)$ and the nominal state $x_{\nt}(t)$, as well as the error between them, and the derived individual bounds from \loneAC . 
Under UC-MPC, the actual states remained consistently close to the nominal states, with the error staying below the computed individual bound.
Similarly, as shown in \cref{fig:ul1}, the adaptive control inputs $u_{\at}(t)$ remained within the theoretical bounds calculated according to \cref{eq:ua-i-bnd-w-Ua-defn}. 

For tube MPC, solving the optimization problem with the additional decision variable associated with the initial condition results in an average computation time of $0.05$~seconds per control step. 
In contrast, the optimization problem in UC-MPC has essentially the same structure as a standard MPC formulated purely on the nominal system. As a result, the optimization is simpler and requires an average computation time of $0.03$~seconds per step. The only additional online computation arises from evaluating the adaptive control input $u_a$, which involves propagating the state predictor and updating the uncertainty estimate through the adaptive law. This procedure requires an average of only $3 \times 10^{-5}$~seconds per step.
These results demonstrate that UC-MPC significantly reduces the online computational burden compared with tube MPC, while the computational overhead introduced by the adaptive component is negligible.

In addition to the previous simulation, we also performed an experiment with only matched uncertainties, i.e., with $w(t)=0$. In this scenario, as illustrated in \cref{fig:tracking_zoom_no_w}, the proposed UC-MPC effectively compensated for the uncertainty, achieving nearly perfect steady-state tracking. Such high-accuracy tracking is unattainable with existing robust or tube MPC methods, and the performance of adaptive MPC approaches would also be constrained when faced with inaccurate parameter estimation.
\begin{figure}[!t]
    \centering
    \includegraphics[width=1\columnwidth]{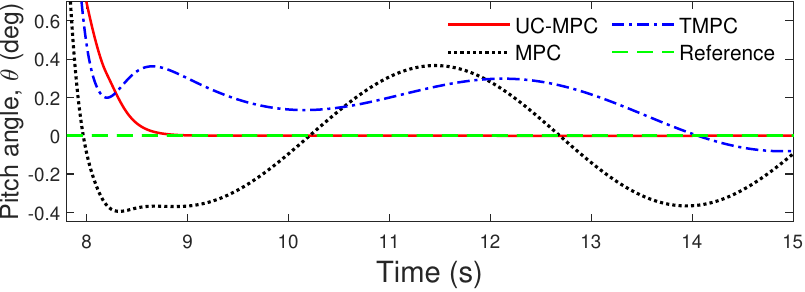} 
    \vspace{-6mm}
    \caption{Zoomed-in view of tracking performance on $\theta$ under MPC, TMPC and UC-MPC when $w(t)=0$.}\label{fig:tracking_zoom_no_w}
        \vspace{-1mm}
\end{figure}

\vspace{-3mm}
\subsubsection{Scaling of $w(t)$}
\vspace{-2mm}
\label{sec:scaling_sim}
To demonstrate the benefit of scaling of $w(t)$, as presented in \cref{sec:scaling of w}, we consider a simplified model for the F16 aircraft represented by \cref{eq:ol-dynamics-f16} with
{\setlength{\arraycolsep}{1pt}
\[ A \!= \!\!
\begin{bmatrix}
    0 & 0.0067& 1.34 \\
    0 & -0.869 & 43.2 \\
    0 & 0.993 & -1.34
    \end{bmatrix} \!,  
    % B \!=\!\! \begin{bmatrix}
    % 0.169 \\
    % -17.3 \\
    % -0.169 
    % \end{bmatrix}\!, B_u\! = \!\!\begin{bmatrix}
    %     0.1500 & 0.0007\\0.0015& -0.0007 \\ -0.0015 &0.0750 
    % \end{bmatrix}, 
    % \\
    B \!\mid\! B_u \!=\!\!
\left[
\begin{array}{c|cc}
0.169  & 1.5\text{e-}1 & 7\text{e-}4 \\
-17.3  & 1.5\text{e-}3  & -7\text{e-}4 \\
-0.169 & -1.5\text{e-}3 & 7.5\text{e-}2
\end{array}
\right]\!,
\]}where the control input $u(t)$ only includes $\delta_e(t)$, and the unmatched uncertainty includes two channels. Now, we consider the following artificial uncertainties:
\begin{align}
    f(t,x) &= -1.44\sin(0.4\pi t)-0.18\alpha^2, \nonumber\\
    w(t)&=[5\sin(0.6\pi t), \cos(0.4\pi t)]^\top,
\end{align}
where obviously different channels in $w(t)$ have different upper bounds, and a scaling of $w(t)$ is necessary to achieve the smallest performance bounds for constraint tightening.
Given the feedback gain matrix $K_x = [1.6238,0.7151,4.8245]$ achieved from the same method as before, the performance bounds from \loneAC~with or without scaling of $w(t)$ are included in \cref{tab:scaling_w}. It is found that when we utilize the scaling of $w(t)$, the calculated performance bounds from \loneAC~ for all control input and state channels are much smaller compared with those without applying scaling of $w(t)$; this validates our discussion in \cref{sec:scaling of w}.

\begin{table}[] 
\centering
\caption{Performance bounds from \loneAC~with and without scaling of $w(t)$.}
\label{tab:scaling_w}
\begin{tabular}{lllllll}
\toprule
                & $\tilde \rho^1$   & $\tilde \rho^2$   & $\tilde \rho^3$   & $\tilde \rho_u^1$   &  &  \\ \midrule
With Scaling    & 1.06 & 1.54  & 0.53 & 9.69 &  &  \\
Without Scaling & 1.35 & 2.15 & 0.74 & 11.63 &  &  \\
                \bottomrule
\end{tabular}
\end{table}

\vspace{-2mm}
\subsection{Soft landing of a spacecraft on an asteroid}
\vspace{-2mm}
Landing a spacecraft on an asteroid presents distinct challenges. First, the asteroid's shape, density, and resulting gravity field are inherently difficult to characterize accurately. Consequently, the nominal system used for controller design often differs significantly from the actual, unknown system. This discrepancy means that controllers designed to meet specific performance and safety requirements for the nominal system may fail to ensure these properties for the real system. Second, the vast distance between the spacecraft and Earth introduces communication delays, making real-time ground-based control infeasible. As a result, the spacecraft must be made autonomous and made to rely entirely on onboard control frameworks. Given these challenges, the UC-MPC framework emerges as a promising solution, offering robust performance and safety guarantees for achieving a soft landing on an asteroid despite significant uncertainties.
\vspace{-3mm}
\subsubsection{System Dynamics}
\vspace{-2mm}
We adopt the system model from \cite{dunham2016constrained} for the soft landing of a spacecraft on an asteroid, where we assume that the spacecraft is equipped with six ideal thrusters, one for each direction, and the asteroid itself is rotating at a constant angular rate about the principal axis with the largest moment of inertia. The equations of motion for the relative position and velocity of the spacecraft in an asteroid fixed frame are given by:
\begin{subequations}\label{eq:original_dynamics}
    \begin{equation}
        \ddot{x} = F_{g,x} + 2n\dot{y}+n^2x+U_x,
    \end{equation}
    \begin{equation}
        \ddot{y} = F_{g,y} - 2n\dot{x}+n^2y+U_y   ,     
    \end{equation}
    \begin{equation}
        \ddot{z} = F_{g,z} +U_z,      
    \end{equation}
\end{subequations}
where $n$ is the rotation rate of the asteroid, $F_{g,x},F_{g,y},F_{g,z}$ are the asteroid gravity forces per unit mass acting on the center of mass of the spacecraft, and $U_x,U_y,U_z$ are the control inputs, i.e., thrust induced accelerations.  
Since the asteroid is not perfectly spherical, the gravity forces can be calculated based on the gravity potential, which is given by:
\begin{align}\label{eq:gravity potential}
    G = -\frac{\mu}{R_0}\sum^\infty_{l=0}\sum^l_{m=0}(C_{l,m}V_{l,m}+S_{l,m}W_{l,m}),
\end{align}
where $\mu$ is the universal gravity constant times the mass of the asteroid, $R_0$ is a reference sphere radius, $C_{l,m}$ and $S_{l,m}$ are the Stokes coefficients, and $V_{l,m}$ and $W_{l,m}$ are functions of position calculated using the recursion scheme in \cite{montenbruck2002satellite}. Based on the gravity potential, we can then calculate the gravity force by: $[F_{g,x}\; F_{g,y}\; F_{g,z}]^\top = - \nabla G$.
To obtain accurate values of the Stokes coefficients, $C_{l,m}$ and $S_{l,m}$, for a small nonuniform body, a direct measurement of spacecraft motion near an asteroid over a long period of time is required. A common alternative is to consider the approximation of the asteroid as a constant density tri-axial ellipsoid with semi-major axes $c\le b\le a$ and shape defined by
$(\frac{x}{a})^2 + (\frac{y}{b})^2+(\frac{z}{c})^2 = 1$.
According to \cite{dunham2016constrained,scheeres2016orbital}, under the constant density ellipsoid assumption, the Stokes coefficients of the asteroid can be calculated out to the fourth order as
\begin{align} \label{eq:stoke coe}
         C_{20} &= \frac{1}{5R_0^2}\frac{c^2-(a^2+b^2)}{2}, C_{22} = \frac{a^2-b^2}{20R^2_0} \\
        C_{40} \!&= \!\frac{15}{7}(C^2_{20}+2C^2_{22}),C_{42} \!= \!\frac{5}{7}C_{20}C_{22}, C_{44} \!=\! \frac{5}{28}C^2_{22} \nonumber
\end{align}
with all $S_{l,m}$ in \eqref{eq:gravity potential} being zero and all other $C_{l,m}$ being zero. 
Using the coefficients calculated in \eqref{eq:stoke coe}, the original nonlinear dynamics \eqref{eq:original_dynamics} can be linearized about a target position equilibrium, e.g., $X_0 = [x_0\; y_0 \;z_0\; 0\; 0\; 0]^\top$, as:
\begin{align}
    \label{eq:linearized dynamics}
    \delta\dot{X}(t) = A\delta X(t) + B\delta U(t),
\end{align}
where $X = [x\; y\; z\; \dot{x}\; \dot{y}\; \dot{z}]^\top$,  $\delta X = X-X_0$, $\delta U = U-U_0$, and $U_0$ is the control necessary for $X_0$ to be an equilibrium.
\vspace{-3mm}
\subsubsection{Controller Design}
\vspace{-2mm}
For the soft landing problem, we model the nominal system \eqref{eq:nominal system new} using the linearized dynamics \eqref{eq:linearized dynamics} derived from the estimated asteroid mass and Stokes coefficients using the constant density ellipsoid approach from \eqref{eq:stoke coe}. The real system is represented by the original nonlinear dynamics \eqref{eq:original_dynamics} with true asteroid mass and Stokes coefficients. Uncertainties arising from linearization, mass estimation, and Stokes coefficient approximation are treated as discrepancies between the nominal and real systems, and \loneAC~ is employed to compensate for these uncertainties. Empirical bounds on these uncertainties, obtained through simulation, are used to tighten constraints in the nominal system, ensuring robust constraint enforcement.

The known gravity model of Eros from \cite{yeomans2000radio} is used to simulate the spacecraft motion as the real system, with $n = 3.3118e\text{-}4$ rads/sec, $\mu = 4.46e\text{-}4\; \text{m}^3/\text{sec}^2$. Stokes coefficients are sourced from NEAR-Shoemaker mission data \cite{miller2002determination}. For the controller design, i,e., nominal system, the constant density ellipsoid is designed with a dimension of $a = 20$ km, $b = 5$ km, and $c = 5$ km and the estimated asteroid mass is designed as $\tilde{\mu} = 4.1e\text{-}4 \text{m}^3/\text{sec}^2$, which is about $10\%$ smaller than the true value.

Following \cite{dunham2016constrained}, the spacecraft landing mission is divided into two distinct phases. During Phase I, the circumnavigation phase, the spacecraft approaches a point near the landing site while avoiding collisions. An ellipsoidal safety boundary around the asteroid ($a = 20$ km, $b = 12$ km, and $c = 7.5$ km) is defined. Collision avoidance is ensured by finding the closest point on the ellipsoid from the spacecraft position, defining a normal vector pointing out of the ellipsoid at the closest point, and constraining the spacecraft motion such that the product between the normal vector and the vector from the closest point to the spacecraft position is positive. The maximum control input is set as $u_{\text{max}} = 1e\text{-}4$ kN/kg and the maximum speed in each direction is required to be smaller than $0.01$ km/s. 
In Phase II, the landing phase, as the spacecraft nears the touchdown point, it follows a narrowing, pyramid-shaped boundary for precise landing. The pyramid-shaped constraints are constructed following the formulation (37) of \cite{dunham2016constrained}. The maximum control input is increased to $u_{\text{max}} = 2e\text{-}3$ kN/kg to ensure a safe landing. For both phases, the MPC cost function is designed as $\delta X^{\top}Q\delta X + \delta U^{\top}R\delta U$ with $Q = diag([1,1,1,100,100,100])$ and $R = 5e8$ during Phase I, and $R = 5e4$ during Phase II. 
\vspace{-3mm}
\subsubsection{Simulation Results}
\vspace{-2mm}
The spacecraft starts at $X_{\text{init}} = [20,-15,0,0,0,0]^\top$ and aims to land at $X_{\text{target}} = [5.86,5.21,1.54,0,0,0]^\top$ on the asteroid. Phase I guides the spacecraft to $X_{\text{target}_1} = [5,12,1,0,0,0]^\top$, after which Phase II manages the descent to $X_{\text{target}}$.
The simulation results of Phase I are presented in \cref{fig:p1_traj,fig:p1_constraint,fig:p1_uncertainty}. \cref{fig:p1_traj} shows three position trajectories: the nominal trajectory generated using vanilla MPC on the nominal linear dynamics \eqref{eq:linearized dynamics}, the UC-MPC trajectory generated with the proposed control approach on the real system dynamics, and the MPC trajectory generated using vanilla MPC on the real system dynamics. As expected, the trajectory of the real system under vanilla MPC deviates significantly from the nominal trajectory due to the mismatch in system dynamics. However, with the proposed UC-MPC framework, the real system trajectory closely tracks the nominal trajectory thanks to the accurate uncertainty estimation and compensation in \cref{fig:p1_uncertainty}. 
In addition, \cref{fig:p1_constraint} shows that the proposed UC-MPC ensures constraints enforcement due to the constraints tightening, while vanilla MPC violates the speed constraints. 
The final steady-state positions further highlight the performance improvement with UC-MPC. The target position for Phase I is $[5,12,1]^\top$. Using UC-MPC, the spacecraft reaches $[5.08,12.73, 1.15]^\top$ with a distance of 0.75 km from the target. In contrast, vanilla MPC results in the spacecraft reaching $[4.32,12.81,1.15]$, with a larger deviation of 1.07 km from the target. This demonstrates the superior accuracy of the UC-MPC framework.
\begin{figure}[!t]
    \centering
    \includegraphics[width=1\columnwidth]{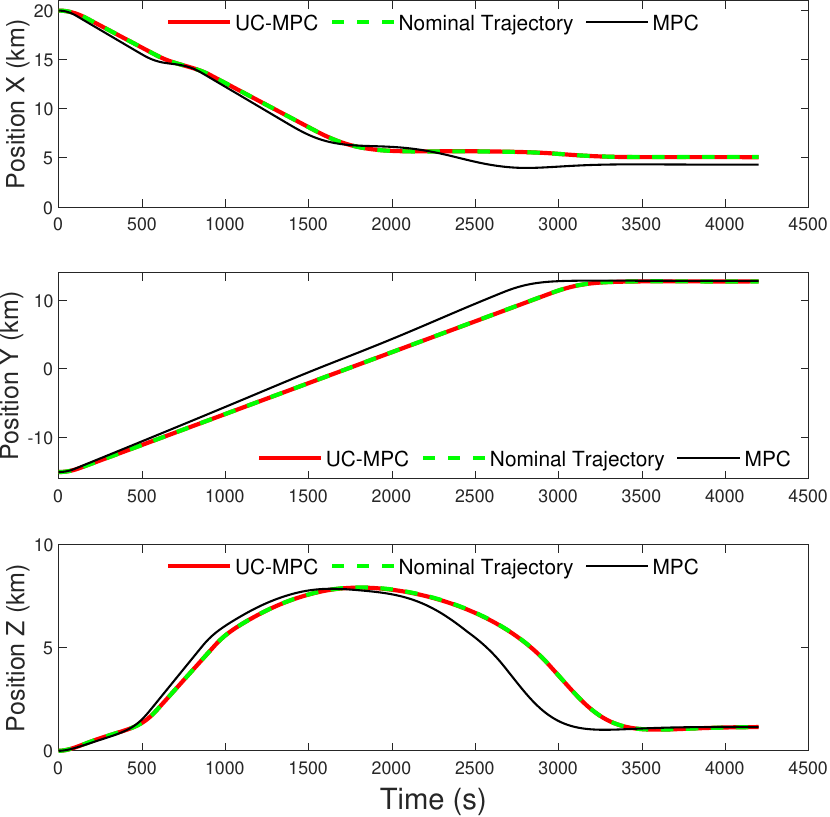} 
    \vspace{-6mm}
    \caption{Phase I position trajectories of the spacecraft. UC-MPC and MPC trajectories are simulated using the real system dynamics. The nominal trajectory is simulated using nominal linear dynamics \eqref{eq:linearized dynamics} with vanilla MPC.}\label{fig:p1_traj}
        \vspace{-1mm}
\end{figure}
\begin{figure}[!t]
    \centering
    \includegraphics[width=1\columnwidth]{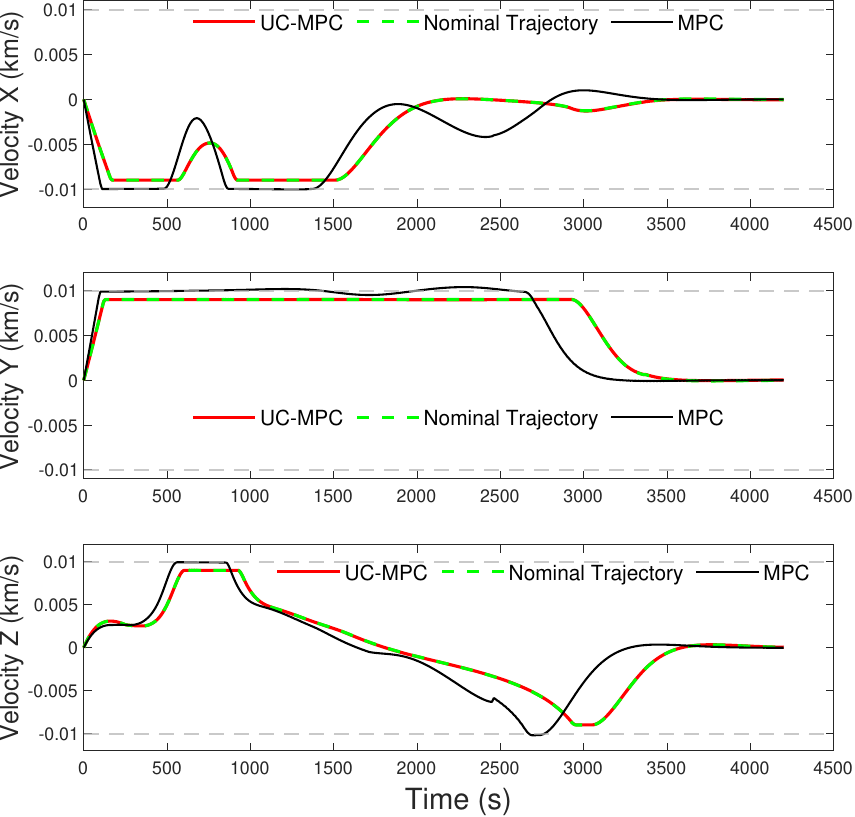} 
    \vspace{-6mm}
    \caption{Phase I trajectories of the constrained velocities under UC-MPC, nominal trajectory, and vanilla MPC. Gray dash-dotted lines illustrate the constraints.} \label{fig:p1_constraint}
        \vspace{-1mm}
\end{figure}

\begin{figure}[!t]
    \centering
    \includegraphics[width=1\columnwidth]{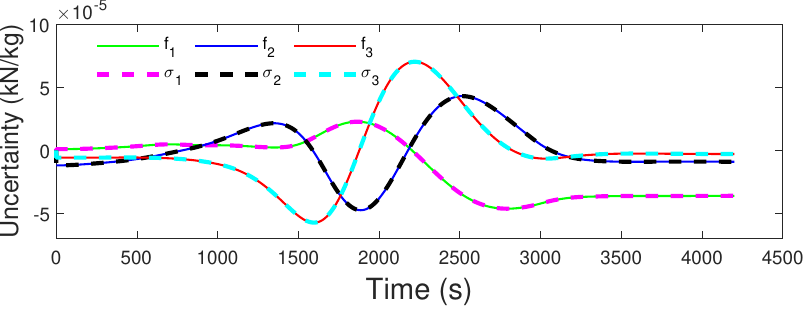} 
    \vspace{-6mm}
    \caption{Phase I actual and estimated uncertainties under UC-MPC. Here $f_i$ represents the true uncertainty, and $\sigma_i$ represents the estimated uncertainty from \loneAC.}\label{fig:p1_uncertainty}
        \vspace{-1mm}
\end{figure}

The simulation results of Phase II are illustrated in \cref{fig:p2_3Dtraj,fig:p2_traj}. Similar to Phase I, the trajectory under vanilla MPC deviates from the nominal trajectory, whereas the UC-MPC trajectory nearly perfectly tracks the nominal trajectory, benefiting from accurate uncertainty estimation. As depicted in  \cref{fig:p2_3Dtraj}, collision avoidance is successfully achieved with the pyramid-shaped constraints during this phase. In terms of final steady-state position, UC-MPC guides the spacecraft to $[5.88,5.21,1.55]^\top$ km just 0.02 km from the target. In comparison, vanilla MPC results in a position of $[5.94,5.21,1.55]^\top$, with a larger error of 0.15 km. These results further emphasize the superior performance of UC-MPC in achieving precise control.

To validate the robustness of the proposed UC-MPC framework, we conducted additional tests in Phase I using varying initial positions while keeping the target position constant. The results, shown in \cref{fig:p1_diff}, demonstrate that the proposed control framework performs reliably with small state error between the real system and the nominal system across different initial conditions, highlighting its robustness to variations in starting positions.

\begin{figure}[!t]
    \centering
    \includegraphics[width=1\columnwidth]{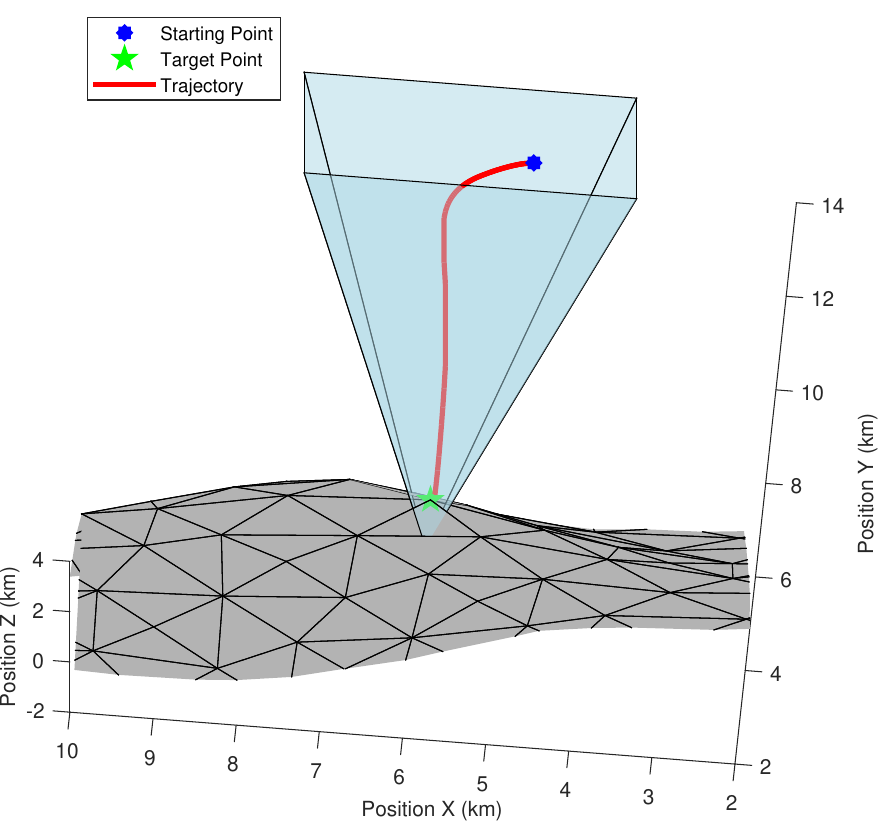} 
    \vspace{-6mm}
    \caption{Phase II 3D position trajectory of UC-MPC with pyramid-shaped constraints.}\label{fig:p2_3Dtraj}
        \vspace{-1mm}
\end{figure}

\begin{figure}[!t]
    \centering
    \includegraphics[width=1\columnwidth]{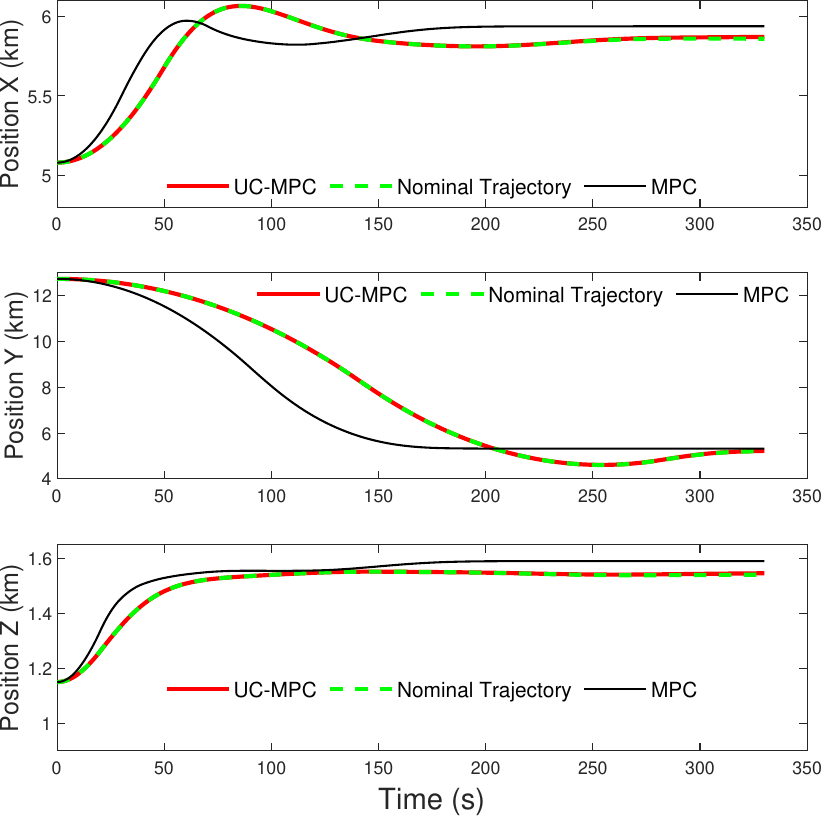} 
    \vspace{-6mm}
    \caption{Phase II position trajectories of the spacecraft. UC-MPC and MPC trajectories are simulated using the real system dynamics. The nominal trajectory is simulated using nominal system dynamics with vanilla MPC.}\label{fig:p2_traj}
        \vspace{-1mm}
\end{figure}

\begin{figure}[!t]
    \centering
    \includegraphics[width=1\columnwidth]{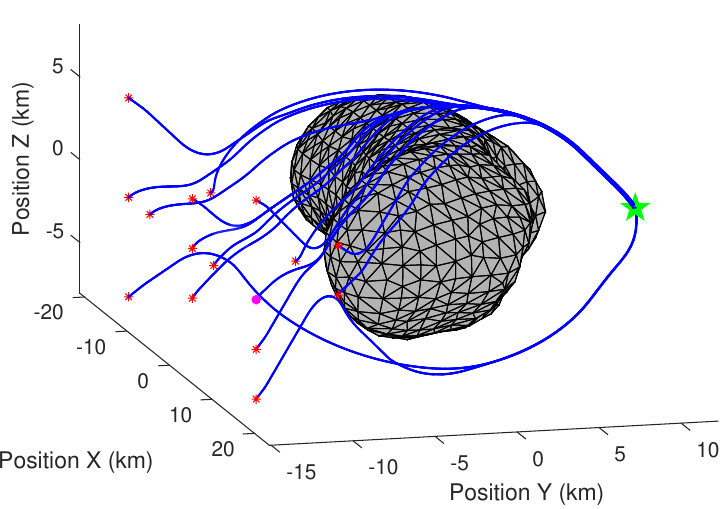} 
     \includegraphics[width=1\columnwidth]{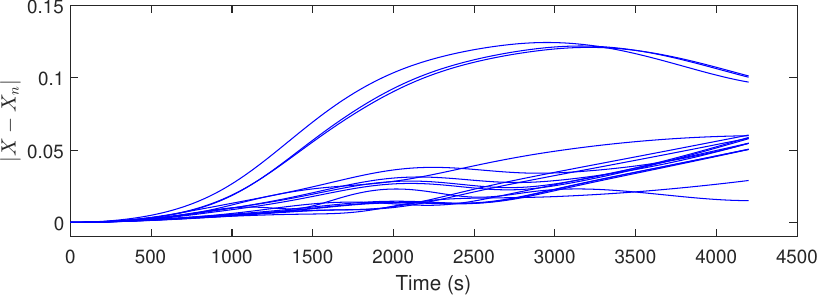} 
    \vspace{-6mm}
    \caption{Trajectories of the spacecraft position under our UC-MPC (top) and   
norm of error between the states of the actual closed-loop system under UC-MPC and the states of the nominal system under nominal MPC (bottom),  i.e., $\|X(t)-X_{\text{n}}(t)\|$, for different initial positions.  The green star indicates the target point, and the red stars are different initial positions of the spacecraft.
    }\label{fig:p1_diff}
        \vspace{-1mm}
\end{figure}

\section{Conclusion}
\vspace{-2mm}
This paper presented UC-MPC, a robust adaptive MPC framework designed to handle linear systems with both matched and unmatched uncertainties while enforcing state and input constraints. The framework incorporates a robust feedback controller designed using LMI methods to reduce the effect of unmatched uncertainties on target output channels, and an adaptive controller to dynamically estimate and compensate for matched uncertainties, achieving uniform performance bounds on the error between the actual system and an ideal uncertainty-free system. These performance bounds are then used to tighten the constraints for the actual system, enabling the design of an MPC problem for the nominal system with adjusted constraints. Simulation results from a flight control problem and a spacecraft soft landing case study demonstrate that the proposed UC-MPC framework outperforms existing methods in terms of performance while maintaining constraint enforcement. 

A limitation of the UC-MPC framework is the assumption that there is sufficient control authority to fully cancel the matched uncertainty through its estimation. Future work will focus on addressing this issue by enabling partial uncertainty cancellation in scenarios where control authority is limited.
%%%%%%%%%%%%%%%%%%%%%%%%%%%%%%%%%%%%%%%%%%%%%%%%%%%%%%%%%%%%%%%%%%%%%%%%%%%%%%%%
\bibliographystyle{plain}        % Include this if you use bibtex 
\bibliography{autosam}           % and a bib file to produce the 
                                 % bibliography (preferred). The
                                 % correct style is generated by
                                 % Elsevier at the time of printing.
\appendix
\section{Proofs}
\vspace{-3mm}
Before presenting the proofs of the lemmas presented in this paper, we first introduce the following lemmas.
\begin{lemma}\label{lem:L1-Linf-relation}
\cite[Lemma 2]{zhao2023integrated}
For a stable proper MIMO system $\mathcal H(s)$ with states $x(t)\in \mbR^n $, inputs $u(t)\in \mbR^m$ and outputs $y(t)\in \mbR^p$, under zero initial states, i.e., $x(0) =0$, we have 
$\linfnormtruc{y}{\tau}\leq \lonenorm{\mathcal H(s)}\linfnormtruc{u}{\tau}$, for any $\tau\geq 0$. Furthermore, for any matrix $\Tau\in \mbR^{q\times p}$,  we have $\linfnorm{\Tau \mathcal H(s)}\leq \infnorm{\Tau}\linfnorm{\mathcal H(s)}$. 
\end{lemma}
\begin{lemma} \label{lem:ref-xr-ur-bnd}
For the closed-loop reference system in \cref{eq:ref-system} subject to \cref{assump:lipschitz-bnd-fi} and the stability condition in  \cref{eq:l1-stability-condition}, we have
\vspace{-4mm}\\
\begin{align}
 \linfnorm{x_{\textup{r}}} &<\rho_r \label{eq:xref-bnd}, \\ \linfnorm{u_{\textup{r}}} &<\rho_{ur}, \label{eq:uref-bnd}
\end{align}where $\rho_r$ is introduced in \cref{eq:l1-stability-condition}, and $\rho_{ur}$ is defined in \cref{eq:rho_ur-defn}.
\end{lemma}
\vspace{-3mm}
\begin{proof}
For notation brevity, we define:
\begin{equation}\label{eq:eta-eta_r-defn}
    \eta(t) \trieq f(t,x(t)),\quad \eta_\rt(t) \trieq f(t,x_{\rt}(t)).
\end{equation}
$\xin(t)$ is defined as the state of the system 
% \begin{equation}\label{eq:xin-dynamics}
    $\dot{x}_{\textup{in}}(t) = A_m \xin(t), \  \xin(0) = x_0,$
% \end{equation}
%signal with Laplace transform 
and we have $\xin(s) \trieq (sI_n-A_m)^{-1} x_0$. 
Let's first rewrite the dynamics of the reference system in \cref{eq:ref-system} in the Laplace domain:
\begin{align}\label{eq:xr-expression-w-G-f}
   \xr(s) = \mcG_{xm}(s)\laplace{f(t,x_{\textup{r}}(t))} \nonumber + \mcH_{xm}(s)u_{\text{opt}}(s)\\+ \mcH_{xu}(s)w(s) + \xin(s).
\end{align}
Considering \cref{eq:rho_in}, $A_m$ is Hurwitz and $\mcX_0$ is compact, we have $\linfnorm{\xin}\leq \rhoin$ according to \cref{lem:L1-Linf-relation}.
As a result, according to \cref{lem:L1-Linf-relation} and because $\linfnorm{w}\leq b_w$, for any $\xi>0$, we have 
\begin{align}\label{eq:x_r_linfnorm_truc_bound}
    \linfnormtruc{x_{\textup{r}}}{\xi}\! &\leq \lonenorm{\mcG_{xm}(s)} \!\linfnormtruc{\eta_\rt}{\xi}\! +\! \lonenorm{\mcH_{xm}(s)}\linfnorm{u_{\text{opt}}}\nonumber\\&+ \! \lonenorm{\mcH_{xu}(s)}\!b_w\!+\!\linfnorm{\xin},
\end{align}
where $\eta_\rt(t)$ is defined in \cref{eq:eta-eta_r-defn}. Assume by contradiction that \cref{eq:xref-bnd} is not true. Since $x_{\textup{r}}(t)$ is continuous and $\infnorm{x_{\textup{r}}(0)}<\rho_r$, there exists a $\tau\!>\!0$ such that 
\begin{equation}
    \infnorm{x_{\textup{r}}(t)}<\rho_r, \ \forall t\in[0,\tau),\ \textup{and}\ \infnorm{x_{\textup{r}}(\tau)}=\rho_r,
\end{equation} 
which implies 
$x_{\rt}(t)\in \Omega(\rho_r)$ for any $t$ in $[0,\tau]$. With \cref{eq:bnd-f} from \cref{assump:lipschitz-bnd-fi}, it follows that
\begin{equation}\label{eq:f-xr-rhor-tau}
    \linfnormtruc{\eta_\rt}{\tau}\leq b_{f,\Omega(\rho_r)}.
\end{equation}
By plugging the inequality above into  \cref{eq:x_r_linfnorm_truc_bound}, we have 
\begin{align}
    \rho_r &\leq \lonenorm{\mcG_{xm}(s)} b_{f,\Omega(\rho_r)}+ \lonenorm{\mcH_{xm}(s)}\linfnorm{u_{\text{opt}}} \nonumber\\&+ \lonenorm{\mcH_{xu}(s)}b_w +  \rho_{\textup{in}},
\end{align}
which contradicts the condition \cref{eq:l1-stability-condition}. Therefore, \cref{eq:xref-bnd} is true. Equation \cref{eq:uref-bnd} immediately follows from \cref{eq:xref-bnd} and \cref{eq:ref-system}. 
\end{proof}

\begin{lemma}\label{lem:xtilde-bnd}
Given the uncertain system \cref{eq:dynamics-uncertain} subject to \cref{assump:lipschitz-bnd-fi}, the state predictor \cref{eq:state-predictor} and the adaptive law \cref{eq:adaptive_law}, if 
\begin{equation}\label{eq:x-u-tau-bnd-assump-in-lemma}
    \linfnormtruc{x}{\tau}\leq \rho, \quad \linfnormtruc{u_{\at}}{\tau}\leq \rho_{u_{\at}},
\end{equation}
with $\rho$ and $\rho_{u_{\at}}$ defined in \cref{eq:rho-u-defn,eq:rho-defn}, respectively, then 
\begin{align}
  \linfnormtruc{\tilx}{\tau} \leq \gamma_0(T). \label{eq:tilx_tau-leq-gamma0}
\end{align}
\end{lemma}
\vspace{-3mm}
\begin{proof}
Based on \cref{eq:x-u-tau-bnd-assump-in-lemma}, we have $x(t)\in \Omega(\rho)$ for any $t$ in $[0,\tau]$. Due to \cref{eq:bnd-f} from \cref{assump:lipschitz-bnd-fi}, it follows that
\begin{equation}\label{eq:f-bnd-in-0-tau}
    \infnorm{f(t,x(t))} =\infnorm{\eta(t)} \leq b_{f,\Omega(\rho)},\quad \forall t\in [0,\tau].
\end{equation} %\infnorm{f(t,x(t))}
From \cref{eq:dynamics-uncertain,eq:state-predictor}, the prediction error dynamics are given by
{\begin{equation}\label{eq:prediction-error}
\begin{aligned}
      \dot{\tilx}(t) &= A_e \tilx(t) + \hsigma(t) - B f(t,x(t))-\bar{B}_u w(t). %\dot{\tilx}(t) &= A_e \tilx(t) + B\left(\hsigma_1(t) - f(t,x(t))\right) + B^\perp \hsigma_2(t). 
      % B^\perp \hsigma_2(t)
\end{aligned}
\end{equation}}
For any $0\leq t <T$ and $i\in\mbZ_0$, due to \cref{eq:prediction-error}, we have
{
\begin{align}
    \tilx(iT+t) &= ~e^{A_et}\tilx(iT)+\int_{iT}^{iT+t} e^{A_e(iT+t-\xi)}%[B \ B^\perp] 
   % \begin{bmatrix} 
   % \hsigma_1(iT) \\
   % \hsigma_2(iT)
   % \end{bmatrix}
\hsigma(iT)   
   d\xi   \nonumber\\&- \int_{iT}^{iT+t} e^{A_e(iT+t-\xi)}\left(B \eta(\xi) + \bar{B}_u w(\xi)  \right) d\xi \nonumber \\ 
    & =  ~e^{A_et}\tilx(iT)+\int_{0}^{t} e^{A_e(t-\xi)}%[B \ B^\perp]
   % \begin{bmatrix}
   % \hsigma_1(iT) \\
   % \hsigma_2(iT)
   % \end{bmatrix}
   \hsigma(iT) 
   d\xi   \nonumber \\&- \int_{0}^{t} e^{A_e(t-\xi)}\left(B \eta(iT+\xi) + \bar{B}_u w(iT+\xi) \right ) d\xi. 
    \label{eq:tilx-iTplust}
\end{align}}
According to the adaptive law \cref{eq:adaptive_law}, the preceding equality implies
\begin{align}
 &\tilx((i+1)T) \!=\nonumber \\ \!&- \!\int_{0}^{T}\! e^{A_e(T-\xi)}\left(B \eta(iT+\xi) + \bar{B}_u w(iT+\xi) \right ) d\xi.
\end{align}
Therefore, considering \cref{eq:f-bnd-in-0-tau}, for any $i\in \mbZ_0$ with $(i+1)T\leq \tau$, we have
\begin{align}
    \infnorm{\tilx((i+1)T)} &\leq \bar \alpha_0(T) b_{f,\Omega(\rho)} +\bar \alpha_1(T) b_w,
\end{align}
where $\bar \alpha_0(T)$ and $\bar \alpha_1(T)$ are defined in \cref{eq:alpha_0-defn} and \cref{eq:alpha_1-defn}, respectively. 
Since $\tilx(0)=0$, 
\begin{align}\label{eq:tilx-iT-bnd}
    \infnorm{\tilx(iT)}\leq \bar \alpha_0(T) b_{f,\Omega(\rho)} +\bar \alpha_1(T) b_w \leq \gamma_0(T), \nonumber \\\; \forall  iT\leq \tau, i\in \mbZ_0.
\end{align}
Now we consider any $t\in(0,T]$ such that $iT+t\leq \tau$ with $i\in \mbZ_0$. From \cref{eq:tilx-iTplust} and the adaptive law \cref{eq:adaptive_law}, we have
\begin{align}
    &\infnorm{\tilx(iT+t)} \leq  \infnorm{e^{A_et}}\infnorm{\tilx(iT)} \hspace{-2cm}\nonumber\\ & +  \int_{0}^{t} \infnorm{e^{A_e(t-\xi)}\Phi^{-1}(T)e^{A_eT}} \infnorm{\tilx(iT)}  d\xi   \nonumber \\
  & +\int_0^t \infnorm{e^{A_e(t-\xi)}B} \infnorm{\eta(iT+\xi)} d\xi \nonumber\\&+ \int_0^t \infnorm{e^{A_e(t-\xi)}\bar{B}_u} \infnorm{w(iT+\xi)} d\xi \nonumber \\
 &\leq   \left(\bar \alpha_2(T)+\bar \alpha_3(T)+1 \right) (\bar\alpha_0(T) b_{f,\Omega(\rho)} +\bar \alpha_1(T)b_w) \nonumber\\ &= \gamma_0(T),  \label{eq:tilx-iT+t-bnd} %+ \bar \alpha_3(T)
\end{align}
where $\bar\alpha_i(T)$ ($i=0,1,2,3$) are defined in \cref{eq:alpha_0-defn,eq:alpha_1-defn,eq:alpha_2-defn}, and the last inequality is partially due to the fact that $\int_0^t \infnorm{e^{A_e(t-\xi)}B}d\xi\leq\int_0^T \infnorm{e^{A_e(T-\xi)}B}\!d\xi\!=\!\bar\alpha_0(T) $ and $\int_0^t \infnorm{e^{A_e(t-\xi)}\bar{B}_u}d\xi\leq\int_0^T \infnorm{e^{A_e(T-\xi)}\bar{B}_u}\!d\xi\!=\!\bar\alpha_1(T) $. Equations \cref{eq:tilx-iT-bnd,eq:tilx-iT+t-bnd} imply \cref{eq:tilx_tau-leq-gamma0}. 
\end{proof}
\vspace{-3mm}
\subsection{Proof of \cref{them:x-xref-bnd}}
\vspace{-3mm}
\begin{proof}
\cref{eq:xref-x-bnd,eq:uref-u-bnd} are proved by contradiction. Assume that \cref{eq:xref-x-bnd} or \cref{eq:uref-u-bnd} do not hold. Since $\infnorm{x_{\textup{r}}(0)-x(0)}=0<\gamma_1$, $\infnorm{u_{\textup{r}}(0)-u_{\at}(0)}=0<\gamma_2$, and $x(t)$, $u_{\at}(t)$, $x_{\textup{r}}(t)$ and $u_{\textup{r}}(t)$ are all continuous, there must exist a time instant $\tau$ such that
\begin{equation}
 \hspace{-2mm}   \infnorm{x_{\textup{r}}(\tau)-x(\tau)} = \gamma_1 \textup{ or }  \infnorm{u_{\textup{r}}(\tau)-u(\tau)} = \gamma_2,
\end{equation}
and
\begin{equation} %\label{eq:xr-x-uar-u-in-0-tau}
   \hspace{-3mm} \infnorm{x_{\textup{r}}(t)\!-\!x(t)}\! <\! \gamma_1, \  \infnorm{u_{\textup{r}}(t)\!-\!u(t)} \!<\! \gamma_2, \ \forall t\in [0,\tau).
\end{equation}
It follows that at least one of the following equalities must hold:
\begin{equation}\label{eq:xr-x-ur-u-linfnorm-tau}
    \linfnormtruc{x_{\textup{r}}-x}{\tau} = \gamma_1, \quad \linfnormtruc{u_{\textup{r}}-u_{\at}}{\tau} = \gamma_2.
\end{equation}
%and furthermore, at least one of the two equations in \cref{eq:xr-x-ur-u-linfnorm-tau} takes the equal sign. 
According to \cref{lem:ref-xr-ur-bnd}, we have $\linfnorm{x_{\textup{r}}} \leq \rho_r<\rho $ and according to \cref{eq:xr-x-ur-u-linfnorm-tau}, we have $\linfnorm{x}\leq \rho_r+\gamma_1 = \rho$. Further considering \cref{eq:lipschitz-cond-f} that results from \cref{assump:lipschitz-bnd-fi}, we achieve
\begin{equation}\label{eq:f-xr-x-tau-bnd}
 \hspace{-1mm}   \infnorm{f(t,x_{\textup{r}}(t))\!-\!f(t,x(t))} \!\leq\! L_{f,\Omega(\rho)}\! \linfnormtruc{x_{\textup{r}}\!-\!x}{\tau}\!, \ \forall t \!\in \![0,\tau].
\end{equation}
The control laws in \cref{eq:l1-control-law} and  \cref{eq:ref-system} indicate that
{
\begin{align}
&u_{\textup{r}}(s) - u_{\at}(s)  = -\mcC(s)\laplace{f(t,x_{\textup{r}})-B^\dagger \hsigma(t)} = \nonumber\\& \mcC(s)\laplace{f(t,x)\!-\!f(t,x_{\textup{r}})} + \mcC(s)(B^\dagger \hsigma(s) \!-\! \laplace{f(t,x)}).\label{eq:omega-ur-u}
\end{align}
From equation \cref{eq:prediction-error}, we have
\begin{equation}\label{eq:hsigma-sigma-s}
    B^\dagger \hsigma(s) - \laplace{f(t,x)}) = B^\dagger(sI_n-A_e)\tilx(s).
\end{equation}}
Considering \cref{eq:dynamics-uncertain}, \cref{eq:l1-control-law} and \cref{eq:hsigma-sigma-s}, we have 
\begin{align}
x(s) \!&= \mcG_{xm}(s) \laplace{f(t,x)} \!+ \!\mcH_{xm}(s)u_{\text{opt}}(s) \!+\! \mcH_{xu}(s)w(s) \nonumber \\ &+\xin(s) 
     - \mcH_{xm}(s)\mcC(s)B^\dagger(sI_n-A_e)\tilx(s),
\end{align}
  which, together with \cref{eq:xr-expression-w-G-f}, implies
  \begin{align}
x_{\textup{r}}(s) - x(s) = \mcG_{xm}(s) \laplace{f(t,x_{\textup{r}})-f(t,x)} \nonumber \\
     +\mcH_{xm}(s)\mcC(s)B^\dagger(sI_n-A_e)\tilx(s).
\end{align}
Therefore, due to \cref{eq:f-xr-x-tau-bnd} and \cref{lem:xtilde-bnd}, we have 
\begin{align}
    \linfnormtruc{x_{\textup{r}}-x}{\tau} &\leq  \lonenorm{\mcG_{xm}} L_{f,\Omega(\rho)} \linfnormtruc{x_{\textup{r}}-x}{\tau}  \\&+\! \lonenorm{\mcH_{xm}(s)\mcC(s)B^\dagger(sI_n-A_e)}\!\gamma_0(T).\nonumber
\end{align}
The preceding equation, together with \cref{eq:l1-stability-condition-Lf}, leads to 
\begin{align}
   \hspace{-2mm} \linfnormtruc{x_{\textup{r}}-x}{\tau} &\!\leq\! \frac{\lonenorm{\mcH_{xm}(s)\mcC(s)B^\dagger(sI_n\!-\!A_e)}}{1-    \lonenorm{\mcG_{xm}}  L_{f,\Omega(\rho)}} \gamma_0(T),
\end{align}
which, together with the sample time constraint \cref{eq:T-constraint}, indicates that 
\begin{equation}\label{eq:xr-x<gamma1}
    \linfnormtruc{x_{\textup{r}}-x}{\tau}  < \gamma_1. 
\end{equation}

On the other hand, it follows from \cref{eq:f-xr-x-tau-bnd,eq:omega-ur-u,eq:hsigma-sigma-s,eq:xr-x<gamma1} that
\begin{align*}
  &\linfnormtruc{ u_{\textup{r}}- u_{\at}}{\tau}   \leq  \lonenorm{\mcC(s)} L_{f,\Omega(\rho)}  \linfnormtruc{x_{\textup{r}}-x}{\tau} \\&+ \lonenorm{\mcC(s)B^\dagger(sI_n-A_e)}\linfnormtruc{\tilx}{\tau} \nonumber \\
    & < \lonenorm{\mcC(s)} L_{f,\Omega(\rho)} \gamma_1 + \lonenorm{\mcC(s)B^\dagger(sI_n-A_e)}\gamma_0(T).
\end{align*}
Further considering the definition in \cref{eq:gamma2-defn}, we have
\begin{equation}\label{eq:ur-u<gamma2}
    \linfnormtruc{ u_{\textup{r}}- u_{\at}}{\tau}<\gamma_2.
\end{equation}
Now, \cref{eq:xr-x<gamma1} and \cref{eq:ur-u<gamma2} contradict the \cref{eq:xr-x-ur-u-linfnorm-tau}, which shows that \cref{eq:xref-x-bnd,eq:uref-u-bnd} holds. The bounds in \cref{eq:x-bnd,eq:ua-bnd} follow directly from  \cref{eq:xref-x-bnd,eq:uref-u-bnd,eq:xref-bnd,eq:uref-bnd} and the definitions of $\rho$ and $\rho_{u_{\at}}$ in \cref{eq:rho-defn,eq:rho-u-defn}.
%, while  \cref{eq:yref-y-bnd} follows from \cref{eq:xref-x-bnd} and  $y_{\textup{r}}(t)-y(t) = C(x_{\textup{r}}(t)-x(t))$. 
The proof is complete. 
\end{proof}
\vspace{-3mm}
\subsection{Proof of \cref{lem:ref-id-bnd}}\vspace{-3mm}
\begin{proof}
Considering \cref{eq:nominal system,eq:ref-system}, we have 
\begin{align}\label{eq:xr-xn-expression}
\hspace{-2mm}   \xr(s)-\xn(s) &= G_{xm}(s) \laplace{f(t,x_{\textup{r}})} + \mcH_{xu}(s)w(s) \nonumber\\&= G_{xm}(s) \laplace{\eta_{\textup{r}}(t)}+ \mcH_{xu}(s)w(s)
\end{align}
From \cref{lem:ref-xr-ur-bnd}, we have $x_{\rt}(t) \in \Omega(\rho_r)$ for any $t\geq0$. Further because
 \cref{eq:bnd-f} that results from \cref{assump:lipschitz-bnd-fi}, it follows that $\linfnorm{\eta_\rt}\leq b_{f,\Omega(\rho_r)} $, which, together with \cref{eq:xr-xn-expression}, leads to \cref{eq:xref-xid-bnd}. 
 %The equation \cref{eq:yref-yid-bnd} follows from the fact that $y_{\textup{r}}(t)-y_{\textup{in}}(t) = C(x_{\textup{r}}(t)-\xn(t))$.
\end{proof}
\vspace{-3mm}
\subsection{Proof of \cref{lem:refine_bnd_xi_w_Txi}}\vspace{-3mm}
\begin{proof}
Given any $\Tau_x^i$ satisfying \cref{eq:Tx-i-cts} with an arbitrary $i\in\Zn$, it follows that $\infnorm{\Tau_x^i}= 1$. As a result, using the transformation \cref{eq:coordinate-trans} and considering \cref{eq:Hxm-Hxv-Gxm-check-defn,eq:check-rhoin-defn} and \cref{lem:L1-Linf-relation}, the following inequalities hold
\begin{subequations}\label{eq:Hxm-Hxv-Gxm-check-original-relation}
\begin{align}
      \hspace{-4mm}   \linfnorm{\mcH_{\check xm}^i(s)}\! & \! \leq \! \infnorm{\Tau_x^i}\! \linfnorm{\mcH_{xm}(s)}\! \!=\! \linfnorm{\mcH_{xm}(s)}\!,  \label{eq:Hxm-check-original-relation} \\
   \hspace{-4mm}   \linfnorm{\mcH_{\check xu}^i(s)}\! & \! \leq \! \infnorm{\Tau_x^i} \! \linfnorm{\mcH_{xu}(s)} \!=\!  \linfnorm{\mcH_{xu}(s)}\!, \label{eq:Hxv-check-original-relation}\\
\hspace{-2mm}  \linfnorm{\mcG_{\check xm}^i(s)}\! & \! \leq \!  \infnorm{\Tau_x^i}  \!\linfnorm{\mcG_{xm}(s)} \!=\!  \linfnorm{\mcG_{xm}(s)}\!, \label{eq:Gxm-check-original-relation} \\
\hspace{-2mm}   \check \rho_{\textup{in}}^i & \! \leq \! \infnorm{\Tau_x^i}  \rhoin \!=\! \rhoin. \label{eq:rhoin-check-original-relation}
%   \\  \check \rho_{\textup{in}}^i & \leq \lonenorm{\mcH_{x v}(s)}\linfnorm{v} + \rho_{\textup{in}}. \label{eq:rhon-check-original-relation}
\end{align}
\end{subequations}
%With the parameter $\mcC(s)$) fixed, let us now analyze the relation between the conditions \cref{eq:l1-stability-condition} for the original system and those for the transformed system. 
The property of $\xr(t)\in \mcX_r$ for any $t\geq0$ from \cref{lem:ref-xr-ur-bnd} and \cref{eq:coordinate-trans} together imply $\check x_{\rt}(t)\in \check \mcX_r$ for any $t\geq0$, where $\check \mcX_r$ is defined via \cref{eq:check-Z-defn}. 
% From \cref{eq:f-checkf-relation,eq:check-Z-defn}, we have 
% \begin{equation}\label{eq:b-checkf-f-equal}
%     b_{\check f, \check \Omega(\rho_r)} = b_{f, \Omega(\rho_r)}.
% \end{equation}
Considering \cref{eq:f-checkf-relation,eq:check-Z-defn}, for any compact set $\mcX_r$, we have
\begin{equation}\label{eq:b-checkf-f-equal}
    b_{\check f, \check \mcX_r} = b_{f, \mcX_r}.
\end{equation}
Suppose that constants $\rho_r$ and $\linfnorm{u_{\text{opt}}}$ satisfy \cref{eq:l1-stability-condition}.  
According to \cref{eq:b-checkf-f-equal,eq:Hxm-Hxv-Gxm-check-original-relation}, with $\check \rho_r^i = \rho_r$ and the same $\linfnorm{u_{\text{opt}}}$, \cref{eq:l1-stability-condition-transformed} is satisfied. %with $\mcG_{\check xm}^i(s)$, $\mcH_{\check xv}^i(s)$ and $\check \rho_{\textup{in}}^i$ that correspond to . 

In addition, if \cref{eq:l1-stability-condition-transformed} holds, through the application of \cref{lem:ref-xr-ur-bnd} to the transformed reference system \cref{eq:ref-system-transformed}, we obtain $\linfnorm{\check x_{\rt}}\leq \check \rho_r^i$, which further implies that $\abs{\check x_{\rt,i}(t)}\leq \check \rho_r^i$ for any $t\geq 0$. Since $\check x_{\rt,i}(t) =x_{\rt,i}(t) $ due to the constraint \cref{eq:Tx-i-cts} on $\Tau_x^i$, we have \cref{eq:xr-i-bnd-from-trans}.
%Similarly, considering $\check x_{\intxt,i}(t) =  x_{\intxt,i}(t)$ and $\check x_{\nt,i}(t) =  x_{\nt,i}(t)$ as well as \cref{eq:check-xn-bnd-w-rhon-defn,eq:check-xin-bnd-check-rhoin}, we have \cref{eq:xin-i-bnd-from-trans,eq:xn-i-bnd-from-trans}. 
Equation \cref{eq:xr-i-bnd-from-trans} is equivalent to $\xr(t)\in\mcX_r$ for any $t\geq 0$, with the re-definition of $\mcX_r$ in \cref{eq:Xr-defn}. 
\begin{comment}
As a result,  the bound on $f(t,\xr(t))$ for any $\xr(t)\in\Omega(\rho_r)$, i.e., $b_{f,\Omega(\rho_r)}$, and the bound on $f(t,\check x_{\rt}(t))$ for any $\check x_{\rt}(t)\in\check \Omega(\rho_r)$, i.e., $b_{\check f,\Omega(\rho_r)}$, 
can now be tightened to be $b_{f,\mcX_r}$ and $b_{\check f,\check \mcX_r}$, respectively.
\end{comment}
Following the proof of \cref{lem:ref-id-bnd}, we are able to obtain $   \left\| {\check x_{\textup{r}} - \check x_{\nt}} \right\|_{{\mathcal L}{_\infty }}  \leq  \lonenorm{\mcG_{ \check xm}} b_{\check f, \check \mcX_r} + \lonenorm{\mcH_{ \check xu}} b_w= \lonenorm{\mcG_{ \check xm} (s)}b_{f, \mcX_r} + \lonenorm{\mcH_{ \check xu}} b_w$, where the equality comes from \cref{eq:b-checkf-f-equal}. Further considering $\check x_{\rt,i}(t) =x_{\rt,i}(t) $ and $\check x_{\nt,i}(t) = x_{\nt,i}(t)$ due to the constraint \cref{eq:Tx-i-cts} on $\Tau_x^i$, we finally have \cref{eq:xri-xni-bnd-from-trans}.
\end{proof}
\vspace{-3mm}
\subsection{Proof of \cref{them:xi-uai-bnd}}\vspace{-3mm}
\begin{proof}
For each $i\in \Zn$, \cref{lem:refine_bnd_xi_w_Txi} implies that %\cref{eq:xn-i-bnd-from-trans-2}, and that 
$\abs{x_{\rt,i}(t)} \leq \check \rho_r^i$ and $\abs{x_{\textup{r,i}}(t)-x_{\nt,i}(t)} \leq \lonenorm{\mcG_{ \check xm}^i(s)}b_{f,\mcX_r} + \lonenorm{\mcH^i_{ \check xu}} b_w$ for all  $t\geq0$. From \cref{them:x-xref-bnd}, it follows that $\abs{x_{\rt,i}(t) - x_{i}(t) } \leq \gamma_1$ for any $t\geq 0$ and any $i\in \Zn$. Thus, \cref{eq:xi-xni-bnd-from-trans-w-tilX-defn} is true. %Further considering the fact that $y_j(t)-y_{\nt,j}(t) = \sum_{i=1}^nC[j,i]\left(x_i(t)-x_{\nt,i}(t)\right)$,  we obtain \cref{eq:yi-yni-bnd-from-trans}.
On the other hand, \cref{them:x-xref-bnd} indicates that $\abs{u_{\rt,j}(t) - u_{\at,j}(t) } \leq \gamma_2$ for any $t\geq 0$. Property \cref{eq:ref-system} and the structure of $\mcC(s)$ in \cref{eq:filter-defn} lead to
\begin{equation}
    u_{\rt,j}(s)=-\mcC_j(s) \laplace{f_j(t,\xr(t))}, \quad \forall j\in\Zm.
\end{equation}
Therefore, given a set $\mcX_r$ such that $\xr(t)\!\in\! \mcX_r$ for any $t\!\geq\! 0$, from \cref{assump:lipschitz-bnd-fi,lem:L1-Linf-relation}, the following property holds 
\begin{equation}\label{eq:uar-i-bnd}
 \abs{u_{\rt,j}(t)} \leq \lonenorm{\mcC_j(s)} b_{f_j, \mcX_r}, \quad \forall t\geq 0, \ \forall j\in\Zm.
\end{equation}
Thus, for any $j\in\Zm$, we have \cref{eq:ua-i-bnd-w-Ua-defn}. Finally, since $u(t) - u_{\nt}(t) = K_x(x(t)-\xn(t)) + u_{\at}(t)$,
we achieve \cref{eq:uj-unj-bnd-w-tilU-defn}. The proof is complete. 
\end{proof}
\end{document}